%% file: main.tex
\documentclass{llncs}

\usepackage{theorem,amsmath,amssymb,multicol,stmaryrd,enumerate,graphicx,color,xspace,fancybox,tikz,paralist,ifthen,verbatim}
\usepackage{gawlitza}
\usepackage{pdfsync}
\usepackage[utf8]{inputenc}
\usepackage[numbers,sort&compress,sectionbib]{natbib}
\usepackage{todonotes}
\usepackage{ifdraft}
\usepackage{marginnote}

\title{%
Numerical Invariants through Convex Relaxation and Max-Strategy Iteration%
\thanks{This work was partially funded by the ANR project ASOPT.}
\ifdraft{\subtitle{Folder: gawlitza/tex/sas2010/journal \\ Date: \today}}{}
}
\author{
  Thomas Martin Gawlitza\inst{1} 
  \and 
  Helmut Seidl\inst{2}}
\institute
{
  VERIMAG, Grenoble, France
  \email{Thomas.Gawlitza@imag.fr}
  \thanks{VERIMAG is a joint laboratory of CNRS, Université Joseph Fourier and Grenoble INP.}
  and
  The University of Sydney, Australia
  \email{gawlitza@it.usyd.edu.au}
  \and
    Technische Universität M\"unchen,
    Institut f\"ur Informatik,
    M\"unchen, Germany
  \email{seidl@in.tum.de}
}
\date{\today}

\newcommand{\citets}[1]{\citet*{#1}}

\renewcommand{\tom}[1]{\ifdraft{\nb{#1}}{}}

\newcommand{\perfect}[1]{{\color{grey}#1}}

\newcommand{\zustand}[2]{\ifthenelse{\equal{p}{#1}}{\perfect{#2}}{#2}}

\newcommand{\np}{\ifdraft{\cleardoublepage}{}}

\newcommand\mymarginpar[1]{%
  \marginnote{\small\bf\sf #1}
}

\newcommand{\per}{\ifdraft{\mymarginpar{P}}{}}
\newcommand{\ok}{\ifdraft{\mymarginpar{O}}{}}

\input{defs}

\begin{document}
\allowdisplaybreaks

\ifdraft{\input{todo}}{}

\maketitle

\input{sdp2}

\bibliographystyle{abbrvnat}
\bibliography{bib}  
\end{document}

%% file: defs.tex
\newcommand{\D}{\mathbb{D}}

\newcommand{\morcave}{morcave\xspace}
\newcommand{\maxmorcave}{$\vee$-morcave\xspace}
\newcommand{\Morcave}{Morcave\xspace}
\newcommand{\cmorcave}{cmorcave\xspace}
\newcommand{\maxcmorcave}{$\vee$-cmorcave\xspace}
\newcommand{\Cmorcave}{Cmorcave\xspace}
\newcommand{\mcave}{mcave\xspace}

\newcommand{\Mcave}{Mcave\xspace}
\newcommand{\cmcave}{cmcave\xspace}
\newcommand{\maxcmcave}{$\vee$-cmcave\xspace}
\newcommand{\Cmcave}{Cmcave\xspace}
\newcommand{\pto}{\rightsquigarrow}

\newcommand{\EvalForMaxAtt}{{\sf EvalForMaxAtt}\xspace}
\newcommand{\EvalForGen}{{\sf EvalForGen}\xspace}
\newcommand{\EvalForCmorcave}{{\sf EvalForCmorcave}\xspace}

\DeclareMathOperator{\dom}{\mathsf{dom}}
\DeclareMathOperator{\fdom}{\mathsf{fdom}}
\DeclareMathOperator{\codom}{\mathsf{codom}}
\DeclareMathOperator{\dcup}{{\dot\cup}}

\newcommand{\suppresol}{\mathsf{suppresol}}

\newcommand{\SDP}{\mathbf{SDP}}
\newcommand{\SDPAaBC}{\SDP_{\A, a, \B, C}}

\newcommand{\LP}{\mathbf{LP}}
\newcommand{\LPAc}{\LP_{A, c}}

\renewcommand{\SR}[1]{S\R^{#1\times #1}}
\newcommand{\SRp}[1]{S\R^{#1\times #1}_+}

\newcommand{\A}{\mathcal{A}}

\newcommand{\Relaxed}{\mathcal{R}}
\newcommand{\Tr}{\mathrm{Tr}}

%% file: sdp2.tex
\allowdisplaybreaks

\begin{abstract}
\input abstract
\end{abstract}

\input{introduction}

\input{basics}

\input{stratimp}
\input para_opt
\input quadratic_zones
\input conc

%% file: abstract.tex
  \per
  In this article we develop a max-strategy improvement algorithm 
  for computing least fixpoints of 
  operators on $\CR^n$ (with $\CR := \R \cup \{ \pm \infty \}$) that
  are point-wise maxima of finitely many monotone and order-concave operators.
  Computing the uniquely determined least fixpoint of 
  such operators is a problem that occurs 
  frequently in the context of 
  numerical program/systems verification/analysis.
  As an example for an application 
  we discuss how our algorithm can
  be applied to compute numerical invariants of programs
  by abstract interpretation based on quadratic templates.

%% file: introduction.tex
\np\section{Introduction}

\subsection{Motivation}

Finding tight invariants for a given program or system 
is crucial for many applications related to program respectively system verification.
Examples include linear recursive filters and numerical integration schemes.
Abstract Interpretation as introduced by 
\citet{DBLP:conf/popl/CousotC77} 
reduces the problem of finding 
tight invariants 
to the problem of finding the uniquely determined least fixpoint of
a monotone operator. 
%
In this article,
we consider the problem of inferring numerical invariants using abstract domains that are based on \emph{templates}.
That is, 
in addition to the program or system we want to analyze, 
a set of templates is given.
These templates are arithmetic expressions in the program/system variables.
The goal then is to compute small safe upper bounds on these templates.
We may, for instance, be interested in computing a safe upper bound on 
the difference $\vx_1 - \vx_2$ of two program/system variables $\vx_1$, $\vx_2$ (at some specified control point of the program).
Examples for template-based numerical invariants include 
\emph{intervals} (upper and lower bounds on the values of the numerical program variables) 
\cite{CouCou76},
\emph{zones} (intervals and additionally upper and lower bounds on the differences of program variables)
\cite{DBLP:conf/rtss/LarsenLPY97,DBLP:conf/eef/Yovine96,DBLP:conf/pado/Mine01}, 
\emph{octagons} 
(zones and additionally upper und lower bounds on the sum of program variables)
\cite{DBLP:conf/wcre/Mine01}, and, more generally,
\emph{linear templates} (also called \emph{template polyhedra}, upper bounds on arbitrary linear functions in the program variables,
where the functions a given a priori)
\cite{DBLP:conf/vmcai/SankaranarayananSM05}.
In this article,
we focus on \emph{quadratic} templates as considered by
\citets{DBLP:conf/esop/AdjeGG10}. 
That is,
a priori, a set of linear and quadratic functions in the program variables (the templates) is given
and we are interested in computing small upper bounds on the values of these functions.
An example for a quadratic template is represented by the quadratic polynomial 
$2 x_1^2 + 3 x_2^2 + 2 x_1 x_2$,
where $x_1$ and $x_2$ are program variables.

When using such a template-based numerical abstract domain,
the problem of finding the minimal inductive invariant, 
that can be expressed in the abstract domain specified by the templates, 
can be recast
as a purely mathematical optimization problem,
where the goal is to minimize a vector 
$(\vx_1,\ldots,\vx_n)$ subject to 
a set of inequalities of the form 
\begin{align}
  \label{eq:intro:motivation}
  \vx_i \geq f(\vx_1,\ldots,\vx_n)
  .
\end{align}

\noindent
Here, $f$ is a \emph{monotone} operator.
The variables $\vx_1,\ldots,\vx_n$ take values in $\R \cup \{\pm\infty\}$.
The variables are representing upper bounds on the values of the templates.
Accordingly, the vector $(\vx_1,\ldots,\vx_n)$ is to be minimized w.r.t.\ the usual component-wise ordering.
Because of the monotonicity of the operators $f$ occurring in the right-hand sides of the inequalities
and the completeness of the linearly ordered set $\R \cup \{\pm\infty\}$,
the fixpoint theorem of Knaster/Tarski ensures the
existence of a uniquely determined least solution.

Computing the least solution of such a constraint system is a difficult task.
Even if we restrict our consideration to the special case of intervals as an abstract domain,
which is, if the program variables are denoted by $x_1,\ldots,x_n$, 
specified by the templates $-x_1,x_1,\ldots,-x_n,x_n$,
the static analysis problem is at least as hard as solving mean payoff games.
The latter problem is a long outstanding problem which is in $\mathsf{NP}$ and in $\mathsf{coNP}$,
but not known to be in $\mathsf{P}$.

A generic way of solving systems of constraints of the form \eqref{eq:intro:motivation}
with right-hand sides that are monotone and variables that range over a complete lattice
is given through the abstract interpretation framework of \citet{DBLP:conf/popl/CousotC77}.
Solving constraint systems in this framework is based on Kleene fixpoint iteration.
However, in our case the lattice has infinite ascending chains.
In this case, termination of the fixpoint iteration
is ensured through an appropriate widening (see \citet{DBLP:conf/popl/CousotC77}).
Widening, however, buys termination for precision. 
Although the lost precision can be partially recovered through a subsequently performed narrowing iteration,
there is no guarantee that the computed result is minimal.

\subsection{Main Contribution}

In this article,
we study the case where the operators
$f$ in the right-hand side of the systems of equations of the form \eqref{eq:intro:motivation}
are not only monotone,
but additionally \emph{order-concave} or even \emph{concave} 
(concavity implies order-concavity, but not vice-versa).
In the static program analysis application we consider in this article,
the end up in this comfortable situation by considering a semi-definite relaxation 
of the abstract semantics.
The concavity of the mappings $f$, however, does not imply 
that the problem can be formulated as a convex optimization problem.
The feasible space of the resulting mathematical optimization problem is normally 
neither order-convex nor order-concave and thus neither convex nor concave.
In consequence, 
convex optimization methods cannot be directly applied.
For the linear case (obtained when using used linear templates),
we solved a long outstanding problem 
---
namely the problem of solving mean payoff games in polynomial time 
--- 
if we would be able to formalize the problem 
through a linear programming problem that can be constructed in polynomial time.

In this article, 
we exploit the fact that the operators $f$ that occur in the right-hand sides of 
the system of inequalities of the form 
\eqref{eq:intro:motivation} we have to solve
are not
only order-concave, but also \emph{monotone}.
In other words: we do not require convexity of the feasible space, 
but we do require monotonicity in addition to the order-concavity.
The main contribution of this article is 
an algorithm for computing least solutions of such systems of inequalities.
The algorithm is based on strategy iteration.
That is,
we consider the process of solving the system of inequality as a game between a maximizer and a minimizer.
The maximizer aims at minimizing the solution, whereas the minimizer aims at minimizing it.
The algorithm iteratively constructs a winning strategy for the maximizer --- a so-called max-strategy.
It uses convex optimization techniques as sub-routines to evaluate parts of the constructed max-strategy.
The concrete convex optimization technique that is used for the evaluation depends on the right-hand sides.
In some cases linear programming is sufficient 
(see \citet{DBLP:conf/csl/GawlitzaS07,DBLP:journals/toplas/GawlitzaS11}), 
In other cases more sophisticated convex optimization techniques are required.
The application we study in this article will require semi-definite programming.
%
%

\per
An important example for monotone and order-concave operators 
are the operators 
that are \emph{monotone} and \emph{affine}.
The class of monotone and order-concave operators is closed under the point-wise infimum operator.
The point-wise infimum of a set of monotone and affine functions, for instance,
is monotone and order-concave.
Another example is the $\sqrt{\phantom{x}}$-operator, 
which is defined by $\sqrt{x} = \sup \; \{ y \in \R \mid y^2 \leq x \}$ for all $x \in \CR$.

\per
An example for a system of inequalities of the class we are considering in this article is the following system of inequalities:
\begin{align}
  \label{eq:intro:ex}
  \vx_1 &\geq \frac12 &
  \vx_1 &\geq \sqrt \vx_2 &
  \vx_2 &\geq \vx_1  &
  \vx_2 &\geq 1 +  \sqrt { \vx_2 - 1 }
\end{align}

\noindent
The uniquely determined least solution of the system \eqref{eq:intro:ex} of inequalities is 
$\vx_1 = \vx_2 = 1$.
We remind the reader again that the important property here is that the right-hand sides of \eqref{eq:intro:ex} are monotone and order-concave.

\per
The least solution of the system \eqref{eq:intro:ex} is also the uniquely determined 
optimal solution of the following convex optimization problem:
\begin{align}
  \max \;\;& \vx_1 + \vx_2 & \text{subject to} &&
  \vx_1 &\leq \sqrt \vx_2 &
  \vx_2 &\leq \vx_1 
\end{align}

\per
\noindent
Observe that the above convex optimization problem is in some sense 
a ``subsystem'' of the 
system \eqref{eq:intro:ex}.
Such a ``subsystem'',
which we will call a max-strategy later on,
is obtained from the 
system \eqref{eq:intro:ex}
by selecting exactly one inequality of the form $\vx_i \geq e_i$ from 
\eqref{eq:intro:ex}
for each variable $\vx_i$ and replacing the relation $\leq$ by the relation $\geq$.
Note that there are exponentially many max-strategies.
The algorithm we present in this article
starts with a max-strategy and assigns a value to it.
It then iteratively improves the current max-strategy 
and assigns a new value to it until the least solution is found.
We utilize the monotonicity and the order-concavity of the right-hand sides to prove that 
our algorithm always terminates with the least solution 
after at most exponentially many improvement steps.
Each improvement step can be executed by solving linearly many convex optimization problems,
each of which can be constructed in linear time.

\per
As a second contribution of this article,
we show how any algorithm for solving such systems of inequalities,
e.g., our max-strategy improvement algorithm, 
can be applied to infer numerical invariants based on quadratic templates.
The method is based on the relaxed abstract semantics introduced by 
\citets{DBLP:conf/esop/AdjeGG10}.

\subsection{Related Work}


\per
\noindent
The most closely related work is  the work of 
\citets{DBLP:conf/esop/AdjeGG10}.
They apply the \emph{min-strategy improvement approach} of \citets{Costan05} 
to the problem of inferring quadratic invariants of programs.
In order to do so,
they introduced the relaxed abstract semantics we are going to use in this article.\footnote{
  \citets{DBLP:conf/esop/AdjeGG10} in fact use the dual version of the relaxed abstract semantics we use in this article. 
  However, this minor difference does not have any practical consequences. 
}
Their method, however,
has several drawbacks compared to the method we present in this article.
The first drawback is that it does not necessarily terminate after finitely many steps.
In addition,
even if it terminates,
the computed solution is not guaranteed to be minimal.
On the other hand, 
their approach also has substantial advantages that are especially important in practice.
Firstly, it can be stopped at any time with a safe over-approximation to the least solution.
Secondly, the computational steps that have to be performed are quite cheap compared the the ones 
we have to perform for the method we propose in this article.
This is caused by the fact that the semi-definite programming problems 
(or in more general cases: convex programming problems)
that have to be solved in each iteration are reasonable small.
We refer to \citets{wingjsc2010}
for a detailed comparison between the max- and the min-strategy approach.


\subsection{Previous Publications}

\per
\noindent
Parts of this work were previously published in the proceedings 
of the Seventeenth International Static Analysis Symposium (SAS 2010).
In contrast to the latter version,
this article contains the full proofs and the precise treatment of infinities.
In order to simplify some argumentations and to deal with infinities,
we modified some definitions quite substentially.
In addition to these improvements, 
we provide a much more detailed study of different classes of order-concave functions 
and the consequences for our max-strategy improvement algorithm.
We do not report on experimental results in this article.
Such reports can be found in the article in the proceedings of 
the Seventeenth International Static Analysis Symposium (SAS 2010).

%
%
%

\subsection{Structure}

\per
\noindent
This article is structured as follows:
Section \ref{s:basics} is dedicated to preliminaries.
We study the class of \emph{monotone} and \emph{order-concave} operators
in Section \ref{s:morcave}.
The results we obtain in Section \ref{s:morcave}
are important to prove the correctness of our max-strategy improvement algorithm.
The method and its correctness proof is presented in Section \ref{s:strat:imp}.
In Section \ref{s:param:opt:rhs}, 
we discuss the important special cases where the right-hand sides 
of the system of inequalities are parametrized convex optimization problems.
This can be used to evaluate strategies more efficiently.
These special cases are important, since they are present especially in the program analysis applications we mainly have in mind.
In Section \ref{s:relaxed:sem},
we finally explain how our methods can be applied to 
a numerical static program analysis based on quadratic templates.
We conclude with Section \ref{s:conc}.

%% file: basics.tex
\np\section{Preliminaries}
\label{s:basics}

\per
\paragraph{Vectors and Matrices}
We denote the $i$-th row (resp.\ $j$-th column) 
of a matrix $A$ by $A_{i\cdot}$ (resp.\ $A_{\cdot j}$).
Accordingly, 
$A_{i \cdot j}$ denotes the component in the
$i$-th row and the $j$-th column.
We also use these notations for vectors and vector valued functions $f : X\to Y^k$,
i.e.,
$f_{i\cdot} (x) = (f(x))_{i\cdot}$ for all $i \in \{1,\ldots,k\}$ and all $x \in X$.

\per
\paragraph{Sets, Functions, and Partial Functions}

We write $A \dcup B$ for the disjoint union of the two sets $A$ and $B$,
i.e.,
$A \dcup B$ stands for $A \cup B$, where we assume that $A \cap B = \emptyset$.
For sets $X$ and $Y$,
$X \to Y$ denotes the set of all functions from $X$ to $Y$, and
$X \pto Y$ denotes the set of all partial functions from $X$ to $Y$.
Note that $X \to Y \subseteq X \pto Y \subseteq X \times Y$.
Accordingly, we apply the set operators $\cup$, $\cap$, and $\setminus$ also to partial functions.
For $X' \subseteq X$,
the restriction $f|_{X'} : X' \to Y$ of a function $f : X \to Y$ to $X'$
is defined by 
$f|_{X'} := f \cap X' \times Y$.
The domain and the codomain of a partial function $f$ are denoted by
$\dom(f)$ and $\codom(f)$, respectively.
For $f : X \to Y$ and $g : X \pto Y$, 
we define $f \oplus g : X \to Y$ by
$
  f \oplus g
  :=
  f|_{X\setminus\dom(g)} \cup g
$.

\per
\paragraph{Partially Ordered Sets}

Let $\D$ be a partially ordered set (partially ordered by the binary relation $\leq$).
Two elements $x, y \in \D$ are called \emph{comparable}
if and only if $x \leq y$ or $y \leq x$.
For all $x \in \D$, 
we set $\D_{\geq x} := \{ y \in \D \mid y \geq x \}$ and $\D_{\leq x} := \{ y \in \D \mid y \leq x \}$.
We denote the \emph{least upper bound} and 
the \emph{greatest lower bound} of 
a set $X \subseteq \D$ by $\bigvee X$ and $\bigwedge X$, 
respectively, provided that it exists.
The least element $\bigvee \emptyset = \bigwedge \D$
(resp.\ the greatest element $\bigwedge \emptyset = \bigvee \D$)
is denoted by $\bot$ (resp.\ $\top$), provided that it exists.
%
%
A subset $C \subseteq \D$ is called a \emph{chain}
if and only if $C$ is linearly ordered by $\leq$,
i.e., it holds $x \leq y$ or $y \leq x$ for all $x,y \in C$.
For every subset $X \subseteq \D$ of a set $\D$ that is partially ordered by $\leq$,
we set $X{\up^\D} := \{ y \in \D \mid \exists x \in X \,.\, x \leq y \}$.
The set $X \subseteq \D$ is called \emph{upward closed w.r.t.\ $\D$}
if and only if $X{\up^\D} = X$.
We omit the reference to $\D$, if $\D$ is clear from the context.


\per
\paragraph{Monotonicity}
Let $\D_1, \D_2$ be partially ordered sets (partially ordered by $\leq$).
A mapping $f : \D_1 \to \D_2$ is called 
\emph{monotone}
if and only if
$f(x) \leq f(y)$ for all $x,y \in \D_1$ with $x \leq y$.
A monotone function $f$ is called 
    \emph{upward-chain-continuous}
    (resp.\ \emph{downward-chain-continuous})
    if and only if 
    $f(\bigvee C) = \bigvee f(C)$
    (resp.\ $f(\bigwedge C) = \bigwedge f(C)$)
    for every non-empty chain $C$ with 
    $\bigvee C \in \dom(f)$
    (resp.\ $\bigwedge C \in \dom(f)$).
    It is  called \emph{chain-continuous}
    if and only if it is \emph{upward-chain-continuous} and 
    \emph{downward-chain-continuous}.

\per
\paragraph{Complete Lattices}
A partially ordered set $\D$
is called a \emph{complete lattice}
if and only if 
$\bigvee X$ and $\bigwedge X$ exist for all $X \subseteq \D$.
If $\D$ is a complete lattice and $x \in \D$,
then the sublattices 
$\D_{\geq x}$ and $\D_{\leq x}$ are also complete lattices.
%
%
On a complete lattice $\D$, 
we define the binary operators $\vee$ and $\wedge$ by 
\begin{align}
  x \vee y := \bigvee \{ x, y \}
  \text{ and }
  x \wedge y := \bigwedge \{ x, y \}
  &&
  \text{for all } x, y \in \D,
\end{align}

\per
\noindent
respectively.
If the complete lattice $\D$ is a 
\emph{complete linearly ordered set} 
(for instance $\CR = \R \cup \{ \pm\infty \}$), 
then $\vee$ is the binary \emph{maximum} operator 
and $\wedge$ the binary \emph{minimum} operator.
For all binary operators $\Box \in \{ \vee, \wedge \}$,
we also consider $x_1 \;\Box\; \cdots \;\Box\; x_k$
as the application of a $k$-ary operator.
This will cause no problems,
since the binary operators $\vee$ and $\wedge$ 
are associative and commutative.


\per
\paragraph{Fixpoints}

Assume that the set $\D$ is partially ordered by $\leq$ and $f: \D\to\D$ is a unary operator on $\D$.
An element $x \in \D$ is called 
\emph{fixpoint} 
(resp.\ \emph{pre-fixpoint}, resp.\ \emph{post-fixpoint})
of $f$
if and only if
$x = f(x)$ (resp.\ $x \leq f(x)$, resp.\ $x \geq f(x)$).
The set of all fixpoints (resp.\ pre-fixpoints, resp.\ post-fixpoints) of $f$ 
is denoted by $\Fix(f)$ (resp.\ $\PreFix(f)$, resp.\ $\PostFix(f)$).
We denote the least (resp.\ greatest) fixpoint of $f$ --- provided that it exists --- by $\mu f$ (resp.\ $\nu f$).
If the partially ordered set $\D$ is a complete lattice and $f$ is monotone,
then the fixpoint theorem of Knaster/Tarski \cite{Tarski55} ensures the existence of $\mu f$ and $\nu f$.
Moreover, 
we have
$\mu f = \bigwedge \PostFix(f)$
and dually
$\nu f = \bigvee \PreFix(f)$.

\per
We write $\mu_{\geq x} f$ (resp.\ $\nu_{\leq x} f$) 
for the least element in the set 
$\Fix(f) \cap \D_{\geq x}$
(resp.\ $\Fix(f) \cap \D_{\leq x}$).
The existence of $\mu_{\geq x} f$ (resp.\ $\nu_{\leq x} f$) 
is ensured if
$\D_{\geq x}$ is a complete lattice and
$f|_{\D_{\geq x}}$ (resp.\ $f|_{\D_{\leq x}}$) is a monotone operator on 
$\D_{\geq x}$ (resp.\ $\D_{\leq x}$),
i.e., if $\D_{\geq x}$ (resp.\ $\D_{\leq x}$) is closed under the operator $f$.
The latter condition is, for instance, fulfilled if
$\D$ is a complete lattice, $f$ is a monotone operator on $\D$, 
and $x$ is a pre-fixpoint (resp.\ post-fixpoint) of $f$.

\per
\paragraph{The Complete Lattice $\CR^n$}
The set of real numbers 
is denoted by $\R$, 
and the complete linearly ordered set $\R \cup \{ \pm\infty \}$
is denoted by $\CR$.
%
Therefore, the set $\CR^n$ is a complete lattice 
that is partially ordered by $\leq$,
where we write $x \leq y$
if and only if $x_{i\cdot} \leq y_{i\cdot}$ for all $i \in \{1,\ldots,n\}$.
As usual,
we write $x < y$ if and only if $x \leq y$ and $x \neq y$.
We write $x \ll y$  if and only if $x_{i\cdot} < y_{i\cdot}$ for all $i \in \{1,\ldots,n\}$.
For $f : \CR^n \pto \CR^m$, 
we set
$
  \fdom(f) := \{ x \in \dom(f) \cap \R^n \mid f(x) \in \R^m \}
  .
$

\per
\paragraph{The Vector Space $\R^n$}
The standard base vectors of the Euclidian vector space $\R^n$ are denoted by $e_1,\ldots,e_n$.
We denote the maximum norm on $\R^n$ by $\norm{\cdot}$,
i.e., 
$\norm x = \max \; \{ \abs{x_{i\cdot}} \mid i \in \{ 1,\ldots,n \} \}$
for all 
$x \in \R^n$.
A vector $x \in \R^n$ with $\norm x = 1$ is called a \emph{unit vector}.

\np\section{\Morcave Operators}
\label{s:morcave}

\per
\noindent
In this section, 
we introduce a notion of order-concavity for functions from the set $\CR^n \to \CR^m$.
We then study the properties of functions that are monotone and order-concave.
The results obtained in this section are used in Section \ref{s:strat:imp} to prove the correctness of our 
max-strategy improvement algorithm.

\subsection{Monotone Operators on $\R^n$}

\per
\noindent
In this subsection, we collect  important properties about monotone operators on $\R^n$.
We start with the following auxiliary lemma:

\begin{lemma}
  \label{l:lambda:d}
  \per
  Let $d, d' \in \R^n$ with $d \gg 0$ and $d' \geq 0$.
  There exist 
  $j \in \{1,\ldots,n\}$ 
  and
  $\lambda, \lambda_1,\ldots,\lambda_n \geq 0$
  such that 
  $\lambda_j = 0$
  and
  $\lambda d = d' + \sum_{i = 1}^n \lambda_i e_i$.
\end{lemma}

\begin{proof}
  \per
  Since $d \gg 0$,
  there exist a $j \in \{1,\ldots,n\}$
  and a $\lambda \geq 0$ 
  such that
  $\lambda d - d' \geq 0$
  and $(\lambda d - d')_{j \cdot} = 0$.
  Thus, there exist 
  $\lambda_1,\ldots,\lambda_n$ with $\lambda_j = 0$
  such that
  $\lambda d - d' = \sum_{i = 1}^n \lambda_i e_i$.
  \qed
\end{proof}

\per
\noindent
We now provide a sufficient criterium for 
a fixpoint $x$ of a monotone partial operator $f$
on $\R^n$ 
for being the greatest pre-fixpoint of $f$.%
\footnote{%
\per%
Note that, 
since $\R^n$ is not a complete lattice,
the greatest pre-fixpoint of $f$ is not necessarily the 
greatest fixpoint of $f$.
The greatest fixpoint of the monotone operators 
$f_1, f_2$ defined by 
$f_1(x) = \frac12x$ and $f_2(x) = 2 x$ for all $x \in \R$,
for instance, is $0$.
This is also the greatest pre-fixpoint of $f_1$, 
but not the greatest pre-fixpoint of $f_2$,
since $f_2$ has no greatest pre-fixpoint.}
Such sufficient criteria are crucial to prove the correctness of our max-strategy improvement algorithm.

\begin{lemma}
\per
  \label{l:pos:richtung:gen}
  Let $f : \R^n \pto \R^n$ be monotone with $\dom(f)$ upward closed,
  $f(x) = x$, and $k \in \N$.
  Assume that, 
  for every $\epsilon > 0$, 
  there exists a unit vector $d_\epsilon \gg 0$ such that
  $f^k(x + \lambda d_\epsilon) \ll x + \lambda d_\epsilon$ 
  for all $\lambda \geq \epsilon$.
  Then, $y \leq x$ for all $y$ with $y \leq f(y)$,
  i.e., $x$ is the greatest pre-fixpoint of $f$.
\end{lemma}

\begin{proof}
  \per
  We show $y \not\leq x \implies y \not\leq f(y)$.
  For that, 
  we first show the following statement:
  \begin{align}
  \label{eq:pos:richtung:1}
    y > x \implies y \not\leq f(y)
  \end{align}
  
  \per
  \noindent
  For that,
  let $y > x$.
  Let $\epsilon := \norm{y - x}$.
  By Lemma \ref{l:lambda:d},
  there exist $\lambda,\lambda_1,\ldots,\lambda_n \geq 0$ 
  with $\lambda_j = 0$ for some $j \in \{1,\ldots,n\}$ 
  such that
  $\overline y := x + \lambda d_\epsilon = y + \sum_{i = 1}^n \lambda_i e_i$ holds.
  We necessarily have $\lambda \geq \epsilon$.
  Using the monotonicity of $f$ 
  and the fact that $f^k(\overline y) \ll \overline y$ holds by assumption,
  we get
  $f^k_{j\cdot}(y) \leq f^k_{j\cdot}(\overline y) < \overline y_{j\cdot} = y_{j\cdot}$.
  Therefore,
  $y \not\leq f(y)$.
  Thus, we have shown \eqref{eq:pos:richtung:1}.
  Now, let $y \not\leq x$.
  Thus, $y' := x \vee y > x$.
  Using \eqref{eq:pos:richtung:1} we get 
  $y ' \not\leq f(y')$.
  For the sake of contradiction assume that
  $y \leq f(y)$ holds.
  Then we get
  $f(y') = f(x \vee y) \geq f(x) \vee f(y) \geq x \vee y = y'$
  ---
  contradiction.
  \qed
\end{proof}

\per
\noindent
In the remainder of this article,
we only use the following corollary of 
Lemma \ref{l:pos:richtung:gen}:

\begin{lemma}
  \label{l:pos:richtung}
  \per
  Let $f : \R^n \pto \R^n$ be monotone with $\dom(f)$ upward closed,
  $f(x) = x$, and $k \in \N$.
  Assume that
  there exists a unit vector $d \gg 0$ such that
  $f^k(x + \lambda d) \ll x + \lambda d$ 
  for all $\lambda > 0$.
  Then, $y \leq x$ for all $y$ with $y \leq f(y)$,
  i.e., $x$ is the greatest pre-fixpoint of $f$.
  \qed
\end{lemma}

\subsection{Monotone and Order-Concave Operators on $\R^n$}


\per%
\noindent
A set $X \subseteq \R^n$ is called \emph{order-convex}
if and only if
$ 
  \lambda x + (1-\lambda) y \in X
$ 
for all comparable $x,y \in X$ and all $\lambda \in [0,1]$.
It is called \emph{convex} if and only if this condition holds 
for all $x,y \in X$.
Every convex set is order-convex, but not vice-versa.
If $n = 1$, then
every order-convex set is convex.
Every upward closed set is order-convex, but not necessarily convex.

\per%
A partial function $f : \R^n \pto \R^m$ 
is called \emph{order-convex} (resp.\ \emph{order-concave}) 
if and only if $\dom(f)$ is order-convex and
\begin{align}
  f(\lambda x + (1-\lambda)y) 
  \leq \text{(resp.\ $\geq$)}\;
  \lambda f(x) + (1-\lambda) f(y)
\end{align}

\per
\noindent
for all comparable $x,y \in \dom(f)$ and all $\lambda \in [0,1]$ 
(cf.\ \citet{OrtegaRheinboldt:book}).
A partial function $f : \R^n \pto \R^m$ 
is called \emph{convex} (resp.\ \emph{concave}) 
if and only if
$\dom(f)$ is convex and 
\begin{align}
f(\lambda x + (1-\lambda)y) 
\leq \text{(resp.\ $\geq$)}\;
\lambda f(x) + (1-\lambda) f(y)
\end{align}

\per
\noindent
for all $x,y \in \dom(f)$ and all $\lambda \in [0,1]$ 
(cf.\ \citet{OrtegaRheinboldt:book}).
Every convex (resp.\ concave) partial function is 
order-convex (resp.\ order-concave),
but not vice-versa.
Note that $f$ is (order-)concave if and only if $-f$ is (order-)convex.
Note also that $f$ is (order-)convex (resp.\ (order-)concave)
if and only if $f_{i\cdot}$ is (order-)convex (resp.\ (order-)concave)
for all $i = 1,\ldots,m$.
If $n = 1$, then every order-convex (resp.\ order-concave) partial function
is convex (resp.\ concave).
  Every order-convex/order-concave partial function 
  is chain-continuous.
  Every convex/concave partial function is continuous.

\per
The set of (order-)convex (resp.\ (order-)concave) partial functions
is not closed under composition.
The functions $f, g$ defined by 
$f(x) = (x-2)^2$ and $g(x) = \frac1x$ for all $x \in \Rpp$,
for instance,
are both convex and thus also order-convex.
However, 
$f \circ g$ with $(f \circ g)(x) = (\frac1x - 2)^2$ for all $x \in \Rpp$
is neither convex nor order-convex.

\per
In contrast to the set of all order-concave partial functions,
the set of all partial functions that are monotone \emph{and} order-concave 
is closed under composition:

\begin{lemma}
  \per
  \label{l:mon:conc:circ}
  Let 
  $f : \R^m \pto \R^n$ and 
  $g : \R^l \pto \R^m$ 
  be monotone and order-convex (resp.\ order-concave).
  Assume that
  $\codom(g) \subseteq \dom(f)$.
  Then $f \circ g$ is monotone and 
  order-convex (resp.\ order-concave).
\end{lemma}

\begin{proof}
  \per
  We assume that $f$ and $g$ are order-convex.
  The other case can be proven dually.
  Let $x,x' \in \dom(g)$ with $x \leq x'$, 
  $y = g(x)$, $y' = g(x')$.
  Since $g$ is monotone, we get $y \leq y'$.
  Since $f$ is monotone, we get 
  $
    (f \circ g)(x) 
    = 
    f(g(x)) 
    = 
    f(y) 
    \leq 
    f(y') 
    =
    f(g(x'))
    =
    (f \circ g)(x') 
  $.
  Hence, $f\circ g$ is monotone.
  
  Let $\lambda \in [0,1]$.
  Then 
  $
    (f\circ g)(\lambda x + (1-\lambda) x')
    =
    f(g(\lambda x + (1-\lambda) x')
    \leq
    f(\lambda g(x) + (1-\lambda) g(x'))
    =
    f(\lambda y + (1-\lambda) y')
    \leq
    \lambda f(y) + (1-\lambda) f ( y' )
    =
    \lambda f(g(x)) + (1-\lambda) f ( g(x') )
    =
    \lambda (f\circ g)(x) + (1-\lambda) (f\circ g)(x') 
  $,
  because $f$ is monotone, and
  $f$ and $g$ are order-convex.
  Hence,
  $f\circ g$ 
  is order-convex.
  \qed
\end{proof}

\subsection{Fixpoints of Monotone and Order-concave Operators on $\R^n$}

\per
\noindent
We now study the fixpoints of 
monotone and order-concave partial operators on $\R^n$.
We are in particular interested in developing 
a simple sufficient criterium for a fixpoint of a monotone and order-concave
partial operator on $\R^n$ for being the greatest pre-fixpoint of this partial operator.
To prepare this,
we first show the following lemma:

\begin{lemma}
  \label{l:eigenschaften:conv}
  \per
  Let $f : \R^n \pto \R^n$ be 
  order-convex
  (resp.\ order-concave).
  Let $x, x^* \in \dom(f)$ 
  with
  $x^* = f(x^*)$, $x \gg \text{(resp.\ $\ll$)} \, f(x)$, $d := x^* - x \gg 0$.
  Then, $x^* + \lambda d \ll \text{(resp.\ $\gg$)} \, f(x^* + \lambda d)$
  for all $\lambda > 0$ with $x^* + \lambda d \in \dom(f)$.
\end{lemma}

\begin{proof}
  \per
  We only consider the case that $f$ is order-convex.
  The proof for the case that $f$ is order-concave can be carried out dually.
  Let $\lambda > 0$.
  Assume for the sake of contradiction
  that there exists some $i \in \{1,\ldots,n\}$ 
  such that 
  $(x^* + \lambda d)_{i\cdot} \geq f_{i\cdot}(x^* + \lambda d)$.
  Since $f_{i\cdot}$ is order-convex 
  and $x_{i\cdot} > f_{i \cdot}(x)$ holds,
  it follows $x^*_{i\cdot} > f_{i\cdot}(x^*)$
  --- contradiction.
  \qed
\end{proof}

%
%
%
%
%

\unnoetig{

\begin{lemma}
  Let $F \subseteq \CR^m\to\CR^n$ 
  be a set of (order-)concave functions.
  Then $g$ defined by 
  $g(x) := \bigwedge \{ f(x) \mid f \in F \}$
  for all $x \in \CR^m$
  is (order-)concave.
\end{lemma}

\begin{proof}
   ??? TODO ???
\end{proof}

\begin{lemma}
  Let $f, g : \CR^m\to\CR^n$ be order-concave,
  and $a \in \R_{\geq 0}$.
  Then $f + g$ and $af$ are order-concave.
\end{lemma}

\begin{proof}
  ??? TODO ???
  \qed
\end{proof}

??? Das folgende könnte man auf den endlichen Bereich erstmal
beschränken, oder ???

}

\per
\noindent
We now use the results obtained so far
to prove the following sufficient criterium 
for a fixpoint of a monotone and order-concave partial operator 
for being the greatest pre-fixpoint.

\begin{lemma}
  \per
  \label{l:element:von:partition:2}
  Let $f : \R^n \pto \R^n$ be monotone and order-concave
  with $\dom(f)$ upward closed.
  Let $x^*$ be a 
  fixpoint of $f$,
  $x$ be a pre-fixpoint of $f$ with $x \ll x^*$,
  and $\mu_{\geq x} f = x^*$.
  Then, $x^*$ is the greatest pre-fixpoint of $f$.
\end{lemma}

\begin{proof}
  \per
  Since $f$ is chain-continuous and $x \ll x^*$ is a pre-fixpoint of $f$,
  there exists some $k \in \N$ such that $x \ll f^k(x)$.
   Let $x'$ be a pre-fixpoint of $f$. 
      Let $d := x^* - x$.
      Note that $d \gg 0$.
      Since 
      $f^k|_{\R^n_{\geq x}} = (f|_{\R^n_{\geq x}})^k$ 
      is monotone and order-concave 
      by Lemma \ref{l:mon:conc:circ}, and
      $x^*$ is a fixpoint of $f^k$ and thus of $f^k|_{\R^n_{\geq x}}$,
      we get
      $
        f^k(x^* + \lambda d) 
        =
        f^k|_{\R^n_{\geq x}}(x^* + \lambda d) 
        \ll 
        x^* + \lambda d
      $
      for all
      $\lambda \in \Rpp$
      by Lemma \ref{l:eigenschaften:conv}.
      Thus, Lemma \ref{l:pos:richtung} 
      gives us 
      $x' \leq x^*$. 
      \qed
\end{proof}




  \begin{example}
    \label{ex:feas:function:1}
    \per
    Let us consider the monotone and concave partial operator $\sqrt\cdot : \R\pto\R$.
    The points $0$ and $1$ are fixpoints of $\sqrt\cdot$,
    since $0 = \sqrt 0$, and $1 = \sqrt 1$.
    Since $\frac12$ is a pre-fixpoint of $\sqrt\cdot$, $\frac12 < 1$, and $\mu_{\geq\frac12} \sqrt\cdot = 1$,
    Lemma \ref{l:element:von:partition:2} implies that 
    $1$ is the greatest pre-fixpoint of $\sqrt\cdot$.
    Observe that for the fixpoint $0$,
    there is no pre-fixpoint $x \in \R$ of $\sqrt\cdot$ with $x < 0$.
    Therefore, Lemma \ref{l:element:von:partition:2} cannot be applied.
    \qed
  \end{example}
  
  \per
  \noindent
  The following example shows that the criterium of 
  Lemma \ref{l:element:von:partition:2} 
  is sufficient,
  but not necessary:
  
  \begin{example}
    \per
    Let $f : \R \to \R$ be defined by
    $f(x) = 0 \wedge x$ for all $x \in \R$.
    Recall that $\wedge$ denotes the minimum operator.
    Then, $0$ is the greatest pre-fixpoint of $f$.
    However,
    there does not exist a $x \in \R$ with $x < 0$ such that
    $\mu_{\geq x} f = 0$,
    since $\mu_{\geq x} f = x$ for all $x \leq 0$.
    Therefore, 
    Lemma \ref{l:element:von:partition:2} 
    cannot be applied to show that $0$ is the greatest pre-fixpoint of $f$.
    \qed
  \end{example}

  \per
  \noindent
  The set $\R^n$ can be identified with the set 
  $\{1,\ldots,n\} \to \R$ which can be identified with the 
  set $\vX \to \R$, whenever $\abs \vX = n$.
  In the remainder of this article,
  we therefore identify the set 
  $(\vX\to\R) \pto (\vX\to\R)$ with 
  the set 
  $\R^n \pto \R^n$
  ---
  provided that $\abs \vX = n$.
  Usually, we use $\vX = \{ \vx_1,\ldots,\vx_n \}$.
  We use one or the other representation depending on 
  which representation is more convenient in the given context.
  
  \per
  Our next goal is to weaken the preconditions of
  Lemma \ref{l:element:von:partition:2},
  i.e., we aim at providing a weaker sufficient criterium
  for a fixpoint of a monotone and order-concave partial operator 
  for being the greatest pre-fixpoint 
  than the one provided by Lemma \ref{l:element:von:partition:2}.
  The weaker sufficient criterium we are going to develop 
  can, for instance, be applied to the following example:

\begin{example}
  \label{ex:two:dim:feasible:function:informal}
  \per
  Let us consider the monotone and order-concave partial operator
  $f : \R^2 \pto \R^2$ defined by
  $
    f(x_1,x_2) 
    := 
    (x_2 + 1 \wedge 0, \; \sqrt x_1 )
  $
  for all $x_1,x_2 \in \R$.
  Then,
  $x^* = (x^*_1,x^*_2) = (0, 0)$ is the greatest pre-fixpoint of $f$.
  In order to prove this,
  assume that $y = (y_1,y_2) > x^*$ is a pre-fixpoint of $f$,
  i.e.,
  $y_1 \leq y_2 + 1$, $y_1 \leq 0$, and $y_2 \leq \sqrt{y_1}$.
  It follows immediately that $y_1 \leq 0$ and thus 
  $y_2 \leq \sqrt{y_1} \leq \sqrt 0 = 0$.
  
  \per
  Lemma \ref{l:element:von:partition:2} is not applicable
  to prove that $x^*$ is the greatest pre-fixpoint of $f$,
  because there is no pre-fixpoint $x$ of $f$ with $x \ll x^*$.
  The situation is even worse:
  there is no $x \in \dom(f)$ with $x \ll x^*$.
  
  \per
  We observe that, locally at $x^* = (0,0)$, the first component $f_{1\cdot}$ of $f$ does 
  not depend on the second argument in the following sense:
  For every $y = (y_1,y_2) \in \R^2$ with $y_1 = x^*_1 = 0$ and $y_2 > x^*_2 = 0$,
  we have $f_{1\cdot}(y) = 0 = f_{1\cdot}(x^*)$.
  The weaker sufficient criterium 
  we develop in the following
  takes this into account.
  That is, we will assume that the set of variables 
  can be partitioned according
  to their dependencies.
  The sufficient criterium of Lemma \ref{l:element:von:partition:2}
  should then hold for each partition.
  In this example this means:
  there exists some 
  $x_1 < x^*_1$ with $x_1 \leq f_{1\cdot}(x_1,x^*_2) = f_{1\cdot}(x_1,0)$
  and $\mu_{\geq x_1} f_{1\cdot}(\cdot,0) = x^*_1 = 0$, 
  and there exists some
  $x_2 < x^*_2$ with $x_2 \leq f_{2\cdot}(x^*_1,x_2) = f_{2\cdot}(0,x_2)$
  and $\mu_{\geq x_2} f_{2\cdot}(0,\cdot) = x^*_2 = 0$.
  We could choose 
  $x_1 = x_2 = -1$,
  for instance.
  \qed
\end{example}

  \per
  \noindent
  In order to derive a sufficient criterium that is weaker than the sufficient 
  criterium of Lemma \ref{l:element:von:partition:2},
  we should,
  as suggested in Example \ref{ex:two:dim:feasible:function:informal},
  partition the variables according to their dependencies.
  In order to define a suitable notion of dependencies,
  let $\vX$ be a set of variables, 
  $f : (\vX\to\R)\pto(\vX\to\R)$ be a monotone partial operator, and
  $\rho : \vX\to\R$.
  For $\vX_1 \dcup \vX_2 = \vX$, 
  we write 
  $\vX_1 \stackrel{f,\rho}\rightarrow \vX_2$
  if and only if
  \begin{enumerate}
    \item
      $\vX_1 = \emptyset$,  
    \item
      $\vX_2 = \emptyset$, or
    \item 
      there exists an $\rho' : \vX_2 \to \R$ 
      with $\rho \oplus \rho' \in \dom(f)$ and $\rho' \ll \rho|_{\vX_2}$ such that 
      $f (\rho \oplus \rho') |_{\vX_1} = f(\rho)|_{\vX_1}$.
  \end{enumerate}

  \per
  \noindent
  Informally spoken,
  $\vX_1 \stackrel{f,\rho}\rightarrow \vX_2$
  states that
  ---
  locally at $\rho$
  ---
  the values of the variables from the set $\vX_1$ 
  \emph{do not depend} on the values of the variables from the set $\vX_2$.
  Dependencies are only admitted in the opposite direction --- from $\vX_1$ to $\vX_2$.
  
  \begin{example}
    \label{ex:two:dim:feasible:function:dep}
    \per
	  Let us again consider the monotone and order-concave partial operator
	  $f : \R^2 \pto \R^2$ 
	  from Example \ref{ex:two:dim:feasible:function:informal} 
	  defined by
	  $
	    f(x_1,x_2) 
	    := 
	    (x_2 + 1 \wedge 0, \; \sqrt x_1 )
	  $
	  for all $x_1,x_2 \in \R$.
	  Note that $f$ is not a total operator,
	  since $\sqrt {x_1}$ and thus $f(x_1,x_2)$ is undefined for all $x_1 < 0$.
	  Moreover, 
	  let 
	  $x := (0, 0 )$.
	  Recall that we  identify the set $\R^2$ with the set $\{\vx_1, \vx_2\} \to \R$.
	  Especially, we identify $x$ with the function 
	  $\{ \vx_1 \mapsto 0,\; \vx_2 \mapsto 0 \}$. 
	  Then, we have $\{ \vx_1 \} \stackrel{f, x}\rightarrow \{\vx_2\}$.
	  That is,
	  locally at $x$,
	  the first component $f_{1\cdot}$ of $f$ does not depend on the second argument.
	  In other words: 
	  locally at $x$,
	  one can strictly decrease the value of the second argument
	  without changing the value of the first component $f_{1\cdot}$ of $f$. 
	  However, 
	  the second component $f_{2\cdot}$ of $f$ may, locally at $x$, 
	  depend on the first argument.
	  In this example, this is actually the case:
	  Locally at $x$, 
	  we cannot decrease the value of the first argument without changing 
	  the value of the second component $f_{2\cdot}$ of $f$.
    \qed
  \end{example}
  
  \per
  \noindent
  If the partial operator $f$ is
  monotone and order-concave,
  then the statement $\vX_1 \stackrel{f,\rho}\rightarrow \vX_2$ also implies 
  that,
  locally at $\rho$,
  the values of the $\vX_1$-components of $f$ do not \emph{increase}
  if the values of the variables from $\vX_2$ increase:
  
  \begin{lemma}
    \label{l:partition:nach:oben:unabhae}
    \per
    Assume that $f : (\vX\to\R) \pto (\vX\to\R)$ 
    is monotone and order-concave.
    If $\vX_1 \stackrel{f,\rho}{\rightarrow} \vX_2$,
    then 
    $(f(\rho \oplus \rho'))|_{\vX_1} = (f(\rho))|_{\vX_1}$ 
    for all 
    $\rho' : \vX_2 \to \R$
    with
    $\rho' \geq \rho|_{\vX_2}$
    and
    $\rho \oplus \rho' \in \dom(f)$.
    \qed
  \end{lemma}
  
  \per
  \noindent
  For $\vX_1 \dcup \cdots \dcup \vX_k = \vX$,
  we write
  $\vX_1 \stackrel{f,\rho}\rightarrow \cdots \stackrel{f,\rho}\rightarrow \vX_k$
  if and only if 
  $k = 1$ 
  or
  $
    \vX_1 \dcup \cdots \dcup \vX_j 
    \stackrel{f,\rho}\rightarrow
    \vX_{j+1} \dcup \cdots \dcup \vX_k
  $
  for all $j \in \{1,\ldots,k-1\}$.

  \per
  Let $\vX$ and $\D$ be sets,
  $f : (\vX \to \D) \pto (\vX \to \D)$,
  and 
  $\vX_1 \dcup \vX_2 = \vX$.
  For $\rho_2 : \vX_2 \to \D$, 
  we define
  $f \leftarrow \rho_2 : (\vX_1\to\D) \pto (\vX_1\to\D)$
  by 
  \begin{align}
    (f \leftarrow \rho_2) (\rho_1)
    &:=
    (f ( \rho_1 \cup \rho_2 )) |_{\vX_1}
    &&
   \text{for all }
   \rho_1 : \vX_1 \to \D.
  \end{align}

\per
\noindent
Informally spoken,
$f \leftarrow \rho_2$ 
is the function that is obtained from $f$ by fixing the values
of the variables from the set $\vX_2$ according to variable assignment $\rho_2$
and afterwards removing all variables from the set $\vX_2$.

\begin{example}
  \label{ex:two:dim:feasible:function:dep:2}
  \per
  Let us again consider the monotone and order-concave partial operator
  $f : \R^2 \pto \R^2$ 
  from 
  Examples \ref{ex:two:dim:feasible:function:informal} 
  and \ref{ex:two:dim:feasible:function:dep} 
  that is defined by
  $
    f(x_1,x_2) 
    := 
    (x_2 + 1 \wedge 0, \; \sqrt x_1 )
  $
  for all $x_1,x_2 \in \R$.
  Let again
  $x := (0, 0)$
  be identified with 
  $x = \{ \vx_1 \mapsto 0, \; \vx_2 \mapsto 0 \}$.
  Then
  $
    (f \leftarrow x|_{\{\vx_2\}})(\rho_1)
    =
    \{ \vx_1 \mapsto 0 \}
  $
  for all $\rho_1 : \{\vx_1\} \to \R$,
  and
  $
    (f \leftarrow x|_{\{\vx_1\}})(\rho_2)
    =
    \{ \vx_2 \to 0 \}
  $
  for all $\rho_2 : \{\vx_2\} \to \R$.
  \qed
\end{example}

\per
\noindent
The weaker sufficient criterium
for a fixpoint of a monotone and order-concave partial operator 
for being the greatest pre-fixpoint of this partial operator 
can now be formalized as follows:

\begin{definition}[Feasibility]
  \label{def:feas}
  \per
  Let $f : (\vX\to\R) \pto (\vX\to\R)$ 
  be monotone and order-concave.
  A fixpoint $\rho^*$ of $f$ is called \emph{feasible} 
  if and only if
  there exist
  $\vX_1 \dcup\cdots\dcup \vX_k = \vX$ 
  with
  $\vX_1 \stackrel{f,\rho^*}\rightarrow \cdots \stackrel{f,\rho^*}\rightarrow \vX_k$
  such that, for each $j \in \{1,\ldots,k\}$,
      there exists some pre-fixpoint $\rho : \vX_j\to\R$
      of $f \leftarrow \rho^*|_{\vX \setminus \vX_j}$
      with $\rho \ll \rho^*|_{\vX_j}$
      such that 
      $\mu_{\geq\rho} (f \leftarrow \rho^*|_{\vX \setminus \vX_j}) = \rho^*|_{\vX_j}$.
  \qed
\end{definition}

\begin{example}
  \label{ex:two:dim:feasible:function}
  \per
  Let us again consider the monotone and order-concave partial operator
  $f : \R^2 \pto \R^2$ from the
  Examples \ref{ex:two:dim:feasible:function:informal},
  \ref{ex:two:dim:feasible:function:dep}, and 
  \ref{ex:two:dim:feasible:function:dep:2} that is
  defined by
  $
    f(x_1,x_2) 
    := 
    (x_2 + 1 \wedge 0, \; \sqrt x_1 )
  $
  for all $x_1,x_2 \in \R$.
  We show that 
  $x := (0, 0 )$ 
  is a feasible fixpoint of $f$.
  From Example \ref{ex:two:dim:feasible:function:informal}, we know that Lemma \ref{l:element:von:partition:2} is not applicable
  to prove that $x$ is the greatest pre-fixpoint.
  Recall that we can identify the set $\R^2$ with the set $\{\vx_1, \vx_2\} \to \R$, and hence 
  $x$ with $\{\vx_1\mapsto 0 ,\; \vx_2\mapsto 0 \}$.
  We have 
  $\{ \vx_1 \} \stackrel{f, x}\rightarrow \{ \vx_2 \}$.
  Moreover, 
  $
    \{ \vx_1 \mapsto -1 \}
    \ll
    x|_{\{\vx_1\}}
  $ 
  is a pre-fixpoint of 
  $f \leftarrow x|_{\{\vx_2\}}$ 
  with
  $ 
    \mu_{\geq \{ \vx_1 \mapsto -1 \}} (f \leftarrow x|_{\{\vx_2\}}) 
    =
    x|_{\{\vx_1\}}
    ,
  $ 
  and
  $
    \{ \vx_2 \mapsto -1 \} 
    \ll 
    x|_{\{\vx_2\}}
  $
  is a pre-fixpoint of $f \leftarrow x|_{\{\vx_1\}}$ with
  $ 
    \mu_{\geq \{ \vx_2 \mapsto -1 \}} (f \leftarrow x|_{\{\vx_1\}}) 
    =
    x|_{\{\vx_2\}}
  $. 
  Thus, $x$ is a feasible fixpoint of $f$.
  \qed  
\end{example}

\per
\noindent
We now show that feasibility is indeed 
sufficient for a fixpoint to be the greatest pre-fixpoint.
Since any fixpoint that fulfills the criterium given by Lemma 
\ref{l:element:von:partition:2}
is feasible, but, as the Examples \ref{ex:two:dim:feasible:function:informal} 
and \ref{ex:two:dim:feasible:function}
show, not vice-versa,
the following lemma is a strict generalization of Lemma
\ref{l:element:von:partition:2}.

\begin{lemma}
  \label{l:feasible:is:greatest}
  \per
  Let $f : (\vX\to\R) \pto (\vX\to\R)$ 
  be monotone and order-concave
  with $\dom(f)$ upward closed, and
  $\rho^*$ be a feasible 
  fixpoint of $f$.
  Then, $\rho^*$ is the greatest 
  pre-fixpoint of $f$.
\end{lemma}

\begin{proof}
  \per
  Since $\rho^*$ is a feasible 
  fixpoint of $f$,
  there exists 
  $\vX_1\dcup\cdots\dcup \vX_k = \vX$
  with 
  $\vX_1 \stackrel{f,\rho^*}\rightarrow \cdots \stackrel{f,\rho^*}\rightarrow \vX_k$
  such that, for each $j \in \{1,\ldots,k\}$,
  there exists some pre-fixpoint 
  $\rho_j$ 
  of $f \leftarrow \rho^*|_{\vX \setminus \vX_j}$
  with $\rho_j \ll \rho^*|_{\vX_j}$
  and
  $\mu_{\geq\rho_j} (f \leftarrow \rho^*|_{\vX \setminus \vX_j}) = \rho^*|_{\vX_j}$.
  Let $\rho'$ be a pre-fixpoint of $f$ with $\rho' \geq \rho^*$
  (it is sufficient to consider this case, since the statement that $\rho''$ is a pre-fixpoint of $f$ 
  implies that $\rho' := \rho^* \vee \rho'' \geq \rho^*$ is also a pre-fixpoint of $f$).
  We show by induction on $j$
  that
  $
    \rho'|_{\vX_{1}\dcup\cdots\dcup \vX_{j}} 
    = 
    \rho^*|_{\vX_{1} \dcup\cdots\dcup \vX_{j}}
  $
  for all 
  $j \in \{1,\ldots,k\}$.
  
  \per
  Firstly, assume that $j = 1$.
      Since 
      $\vX_1 \stackrel{f,\rho^*}\rightarrow \vX_2 \dcup\cdots\dcup \vX_k$,
      Lemma \ref{l:partition:nach:oben:unabhae} gives us 
      $
        \rho^*|_{\vX_1}
        =
        (f(\rho^*))|_{\vX_1}
        =
        (f \leftarrow \rho^*|_{\vX\setminus \vX_1})(\rho^*|_{\vX_1})
        =
        (f \leftarrow \rho'|_{\vX\setminus \vX_1})(\rho^*|_{\vX_1})
      $.
      Using the monotonicity we thus get
      $\mu_{\geq\rho_1} (f \leftarrow \rho'|_{\vX \setminus \vX_1}) = \rho^*|_{\vX_1}$.
      Hence,
      Lemma \ref{l:element:von:partition:2}
      gives us
      that 
      $\rho^*|_{\vX_1}$ 
      is the greatest pre-fixpoint of 
      $f \leftarrow \rho'|_{\vX \setminus \vX_1}$.
      Thus, $\rho'|_{\vX_1} = \rho^*|_{\vX_1}$.

  \per      
  Now, assume that $j \in \{2,\ldots,k\}$
  and 
  $
    \rho'|_{\vX_{1} \dcup\cdots\dcup \vX_{j-1}}
    =
    \rho^*|_{\vX_{1} \dcup\cdots\dcup \vX_{j-1}}
  $.
  It remains to show that 
  $
    \rho'|_{\vX_{j}}
    =
    \rho^*|_{\vX_{j}}
  $.
      Since 
      $
        \vX_1 \dcup\cdots\dcup \vX_{j}
        \stackrel{f,\rho^*}\rightarrow 
        \vX_{j+1} \dcup\cdots\dcup \vX_k
      $
      and 
  $
    \rho'|_{\vX_{1} \dcup\cdots\dcup \vX_{j-1}}
    =
    \rho^*|_{\vX_{1} \dcup\cdots\dcup \vX_{j-1}}
  $,
      Lemma \ref{l:partition:nach:oben:unabhae} gives us that
      $
        \rho^*|_{\vX_j}
        =
        (f(\rho^*))|_{\vX_j}
        =
        (f \leftarrow \rho^*|_{\vX\setminus \vX_j})(\rho^*|_{\vX_j})
        =
        (f \leftarrow \rho'|_{\vX\setminus \vX_j})(\rho^*|_{\vX_j})
      $.
      By monotonicity, we thus get
      $\mu_{\geq\rho_j} (f \leftarrow \rho'|_{\vX \setminus \vX_j}) = \rho^*|_{\vX_j}$.
      Hence, 
      Lemma \ref{l:element:von:partition:2}
      gives us
      that 
      $\rho^*|_{\vX_j}$ 
      is the greatest pre-fixpoint of 
      $(f \leftarrow \rho'|_{\vX \setminus \vX_j})$.
      Hence $\rho'|_{\vX_j} = \rho^*|_{\vX_j}$.
      Thus,
      we get 
      $
        \rho'|_{\vX_{1} \dcup\cdots\dcup \vX_{j}}
        =
        \rho^*|_{\vX_{1} \dcup\cdots\dcup \vX_{j}}
      $.
  \qed
\end{proof}


\subsection{\Morcave Operators on $\CR^n$}

\per
\noindent
We now study total operators on $\CR$ that are monotone and order-concave.
For that, we firstly extend the notion of order-concavity that is defined for
partial operators on $\R$ to total operators on $\CR$.
Before doing so, we start with the following observation:

\begin{lemma}
  \per
  \label{l:f:mon:dann:dom:order-concave}
  Let $f : \CR^n \to \CR^m$ be monotone.
  Then, $\fdom(f)$ is order-convex.
\end{lemma}

\begin{proof}
  \per
  Let $x, y \in \fdom(f)$ with $x \leq y$ and $\lambda \in [0,1]$.
  Because of the monotonicity of $f$, 
  we get $\neginfty < f(x) \leq f(\lambda x + (1-\lambda) y) \leq f(y) < \infty$.
  Hence, $\lambda x + (1-\lambda) y \in \fdom(f)$.
  This proves the statement.
  \qed 
\end{proof}

\per
\noindent
We extend the notion of (order-)convexity/(order-)concavity from $\R^n \pto \R$ to
$\CR^n \to \CR$ as follows:
let $f : \CR^n \to \CR$, 
and $I : \{ 1,\ldots,n \} \to \{ \neginfty, \mathsf{id}, \infty \}$
be a mapping.
Here, 
$\neginfty$ denotes the function that assigns $\neginfty$ to every argument,
$\mathsf{id}$ denotes the identity function,
and $\infty$ denotes the function that assigns $\infty$ to every argument.
We define the mapping $f^{(I)} : \CR^{n} \to \CR$ by 
\begin{align}
  f^{(I)} (x) 
  &:= f(I(1)(x_{1\cdot}),\ldots,I(n)(x_{n\cdot})) 
  &&
  \text{for all } x \in \CR^n
  .
\end{align}

\per
\noindent
  A function $f : \CR^n \to \CR$ is called \emph{(order-)concave}
  if and only if
  the following conditions are fulfilled 
  for all mappings $I : \{ 1,\ldots,n \} \to \{ \neginfty, \mathsf{id}, \infty \}$:
  \begin{enumerate}
    \item
      $\fdom(f^{(I)})$ is (order-)convex.
    \item
      $f^{(I)}|_{\fdom(f^{(I)})}$ is (order-)concave.
    \item
      If 
      $\fdom(f^{(I)}) \neq \emptyset$,
      then
      $f^{(I)}(x) < \infty$ for all $x \in \R^n$.
  \end{enumerate}

\per
\noindent
Note that, by Lemma \ref{l:f:mon:dann:dom:order-concave},
condition 1 is fulfilled for every monotone function $f : \CR^n\to\CR$ and 
every mapping $I : \{ 1,\ldots,n \} \to \{ \neginfty, \mathsf{id}, \infty \}$.
A monotone operator is order-concave if and only if the following conditions are fulfilled 
for all mappings $I : \{1,\ldots,n\} \to \{\neginfty,\mathsf{id},\infty\}$:
\begin{enumerate}
  \item
      $\fdom(f^{(I)})$ is upward closed w.r.t.\ $\R^n$.
  \item
      $f^{(I)}|_{\fdom(f^{(I)})}$ is order-concave.    
\end{enumerate}

\per
\noindent
In order to get more familiar with the above definition,
we consider a few examples of order-concave operators on $\CR$:

\begin{example}
  \per
  We consider the operators 
  $f : \CR^2\to\CR$ 
  and 
  $g : \CR^2\to\CR$   
  that are defined by
  \begin{align}
    f (x_1, x_2)
    &:=
    \sqrt {x_1} 
    , 
    &
    g (x_1, x_2)
    &:=
    \begin{cases}
    \sqrt {x_1} & \text{if } x_2 < \infty \\
    x_1^2 & \text{if } x_2 = \infty
    \end{cases}
    &&
    \text{for all } x_1,x_2 \in \CR
    .
  \end{align}
  
  \per
  \noindent
  Then,
  $f|_{\R^2} = g|_{\R^2} = \{ (x_1,x_2) \mapsto \sqrt{x_1} \mid x_1, x_2 \in \R \}$ 
  is a 
  monotone and concave operator on the convex set 
  $\fdom(f) = \fdom(g) = \Rp \times \R$.
  Nevertheless,
  $f$ is monotone and order-concave 
  whereas $g$ is neither monotone nor order-concave.
  In order to show that $g$ is not order-concave,
  let $I : \{1,2\} \to \{\neginfty,\mathsf{id},\infty\}$ be defined by
  $I(1) = \mathsf{id}$ and $I(2) = \infty$.
  Then, $g^{(I)}(x_1,x_2) = x_1^2$ for all $x_1,x_2 \in \CR$.
  Hence, $\fdom(g^{(I)}) = \R^2$.
  Obviously,
  $g^{(I)}|_{\R^2}$ is not order-concave.
  Therefore, $g$ is not order-concave.
  
  \per
  Another example for a monotone and order-concave 
  operator is the function
  $h : \CR^2\to\CR$ 
  defined by
  \begin{align}
    h (x_1, x_2)
    &=
    \begin{cases}
    \sqrt {x_1} & \text{if } x_2 < \infty \\
    \sqrt {x_1} + 1 & \text{if } x_2 = \infty
    \end{cases}
    &&
    \text{for all } x_1,x_2 \in \CR
    .
  \end{align}  
  
  \per
  \noindent
  Although $h$ is an order-concave operator on $\CR$,
  it is not upward-chain-continuous, 
  since,
  for $C = \{ (0,i) \mid i \in \R \}$,
  we have
  $h(\bigvee C) = h(0,\infty) = 1 > 0 = \bigvee \{ 0 \} = \bigvee h(C)$.
  We study different classes of monotone and order-concave functions
  in the remainder of this article.
  \qed
\end{example}

\per
\noindent
A mapping $f : \CR^n \to \CR^m$ is called 
\emph{(order-)concave}
if and only if $f_{i\cdot}$ is (order-)concave for all $i \in \{1,\ldots,m\}$.
  A mapping $f : \CR^n \to \CR^m$ is called 
  \emph{(order-)convex} 
  if and only if
  $-f$ is (order-)concave.

\per
One property we expect from the set of all order-concave functions 
from $\CR^n$ in $\CR^m$ is that 
it is closed under the point-wise infimum operation.
This is indeed the case:

\begin{lemma}
  \per
  Let $\mathcal F$ be a set of (order-)concave functions from $\CR^n$ in $\CR^m$.
  The function $g : \CR^n\to\CR^m$ defined by
  $g(x) := \bigwedge \{ f(x) \mid f \in \mathcal F \}$ for all $x \in \CR^n$
  is (order-)concave.
\end{lemma}

\begin{proof}
  \per
  The statement can be proven straightforwardly.
  Note that $g(x) = (\infty,\ldots, \allowbreak \infty)$ for all $x \in \CR^n$
  if $\mathcal F = \emptyset$.
  In this case, $g$ is concave.
  \qed
\end{proof}

\per
\noindent
Monotone and order-concave functions play a central role in the 
remainder of this article.
For the sake of simplicity,
we give names to important classes of monotone and order-concave
functions:

\begin{definition}[\Morcave, \Mcave, \Cmorcave, and \Cmcave Functions]
  \per
  A mapping $f : \CR^n\to\CR^m$ is called 
  \emph\morcave if and only if it is monotone and order-concave.
  It is called \emph\mcave if and only if it is monotone and concave.
  It is called \emph\cmorcave (resp.\ \emph{\cmcave}\!\!\!\,)
  if and only if it is \morcave (resp.\ \mcave)
  and 
  $f^{(I)}_{i\cdot}$ 
  is upward-chain-continuous on 
  $\{ x \in \CR^n \mid f^{(I)}_{i\cdot}(x) > \neginfty \}$
  for all $I : \{1,\ldots,n\} \to \{\neginfty,\mathsf{id},\infty\}$ and all $i \in \{1,\ldots,n\}$.
  \qed
\end{definition}

\begin{example}
  \per
  Figure \ref{fig:mon_o_conc} shows the graph of a  
  \morcave function $f : \CR^2 \to \CR$.
  \qed
\end{example}

\begin{figure}
  \centering\scalebox{0.75}{\includegraphics{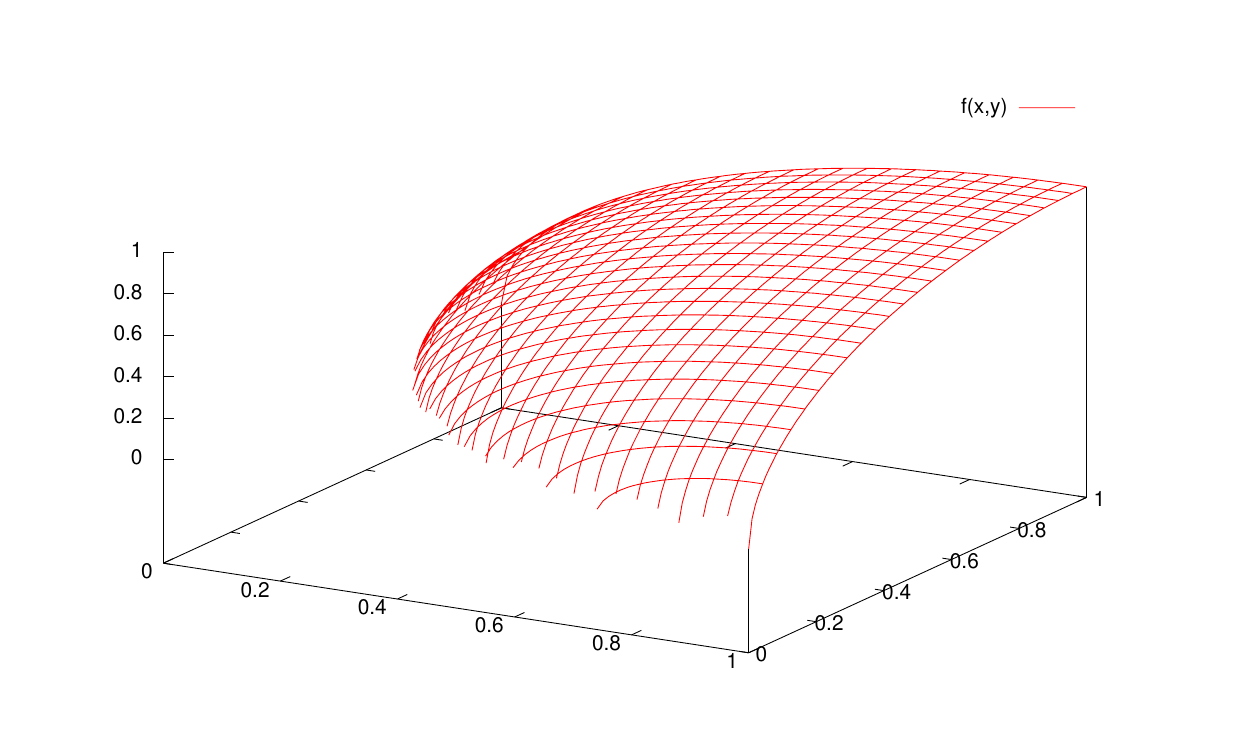}}
%
  \caption{Graph of a \morcave operator $f : \CR^2 \to \CR$. }
  \label{fig:mon_o_conc}
\end{figure}


\per
\noindent
An important \cmcave operator 
for our applications is the operator $\wedge$ on $\CR^n$:

\begin{lemma}
  \label{l:elementare:conv:funcs:1}
  \per
  The operator $\vee$ on $\CR^n$ is monotone and convex, but not order-concave.
  The operator $\wedge$ on $\CR^n$ is \cmcave, but not order-convex.
  \qed
\end{lemma}

\per
\noindent
Next, 
we extend the definition of \emph{affine} functions from $\R^n\to\R^m$ 
to a definition of \emph{affine} functions from $\CR^n\to\CR^m$.

\begin{definition}[Affine Functions]
\per
	A function $f : \R^n\to\R^m$ is called \emph{affine}
	if and only if there exist some $A \in \R^{m \times n}$ and some $b \in \R^m$ 
	such that $f(x) = Ax + b$ for all $x \in \R^n$.
	A function $f : \CR^n \to \CR^m$ 
	is called \emph{affine} if and only if there exist some $A \in \R^{m \times n}$ and some $b \in \CR^m$ 
	such that $f(x) = Ax + b$ for all $x \in \CR^n$.
	\qed
\end{definition}

\per
\noindent
In the above definition and throughout this article, 
we use the convention that $\neginfty + \infty = \neginfty$.
Observe that an affine function $f$ with $f(x) = Ax+b$ is 
monotone, whenever
all entries of the matrix $A$ are non-negative.

\begin{lemma}
  \per
  \label{l:elementare:conv:funcs}
  Every affine function $f : \R^n \to \R^m$
  is concave and convex.
  Every monotone and affine function 
  $f : \CR^n \to \CR^m$ 
  is 
  \cmcave.
  \qed
\end{lemma}

\per
\noindent
In contrast to the class of monotone and order-concave operators on $\R$,
the class of \morcave operators on $\CR$ is not closed under functional composition,
as the following example shows:

\begin{example}
  \per
  We consider the functions 
  $f : \CR\to\CR$ and $g : \CR\to\CR$ defined by
  \begin{align}
    f(x) 
    &:=
    \begin{cases} 
      0 & \text{if } x = \neginfty \\
      1 & \text{if } x > \neginfty
    \end{cases}
    &
    g(x) 
    &:= 
    \begin{cases}
      \neginfty & \text{if } x < 0 \\
      0             & \text{if } x \geq 0 
    \end{cases}
    && 
    \text{for all } x \in \CR
    .
  \end{align}
  
  \per
  \noindent
  The functions $f$ and $g$ are both \morcave\ --- even \cmcave.
  However, 
  observe that
  \begin{align}
    (f\circ g)(x) = f(g(x))
    &:=
    \begin{cases} 
      0 & \text{if } x < 0 \\
      1 & \text{if } x \geq 0
    \end{cases}
    && 
    \text{for all } x \in \CR
    .
  \end{align}
  
  \per
  \noindent
  Then,
  $f \circ g$ is monotone, but not order-concave.
  \qed
\end{example}

\per
\noindent
As we will see,
the composition $f \circ g$ of two \morcave operators $f$ and $g$
is again \morcave,
if $f$ is additionally strict in the following sense:
a function $f : \CR^n\to\CR$ is called \emph{strict} if and only if
$f(x) = \neginfty$ 
for all $x \in \CR^n$ with $x_{k\cdot} = \neginfty$ for some $k \in \{1,\ldots,n\}$.

\begin{lemma}
  \per
  Let $f : \CR^m \to \CR$ and $g : \CR^n \to \CR^m$ be
  \morcave.
  Assume additionally that $f$ is strict.
  Then $f \circ g$ is \morcave.
\end{lemma}

\begin{proof}
  \per
  Since $f$ and $g$ are monotone,
  $f \circ g$ is also monotone.
  In order to show that $f \circ g$ is order-concave,
  let 
  $I : \{ 1,\ldots,n \} \to \{ \neginfty, \mathsf{id}, \infty \}$ and
  $h := (f\circ g)^{(I)}$.
  
 \begin{enumerate}
 \item
  \per
  The set $\fdom(h)$ is order-convex by 
  Lemma \ref{l:f:mon:dann:dom:order-concave},
  since $h$ is monotone.
  
 \item
  \per
  Let $x, y \in \fdom(h)$ with $x \leq y$, $\lambda \in [0,1]$, and 
  $z := \lambda x + (1-\lambda)y$.
  Moreover, 
  let $x' := g^{(I)}(x)$, $y' := g^{(I)}(y)$, and $z' := g^{(I)}(z)$.
  The strictness of $f$ implies that $z' \gg (\neginfty,\ldots,\neginfty)$.
  Since $g^{(I)}$ is monotone,
  we get $x' \leq y'$.
  We define 
  $I' : \{1,\ldots,m\}\to\{\neginfty,\mathsf{id},\infty\}$ 
  by
  \begin{align*}
    I'(k) 
    &=
    \begin{cases}
      \mathsf{id}      & \text{if } z'_{k\cdot} \in \R \\
      \mathsf{\infty} & \text{if } z'_{k\cdot} = \infty 
    \end{cases}
    && 
    \text{for all } k \in \{1,\ldots,m\}
    .
  \end{align*}
  
  \per
  \noindent
  We get:
  \begin{align*}
    h(z) 
    &=
    f(g^{(I)}(z))
    \\&=
    f^{(I')}(g^{(I)}(z))
    \\&\geq
    f^{(I')}(\lambda g^{(I)}(x) + (1-\lambda) g^{(I)}(y))
    \\&\qquad\qquad\qquad \text{(Monotonicity, Order-Concavity)}
    \\&=
    f^{(I')}(\lambda x' + (1-\lambda) y')
    \\&\geq
    \lambda f^{(I')}( x' ) + (1-\lambda) f^{(I')}(y')
    & \text{(Order-Concavity)}
    \\&=
    \lambda f^{(I')}( g^{(I)}(x) ) + (1-\lambda) f^{(I')}( g^{(I)}(y) )
    \\&\geq
    \lambda f( g^{(I)}(x) ) + (1-\lambda) f( g^{(I)}(y) )
    & \text{($f \leq f^{(I')}$)}
    \\&=
    \lambda h(x) + (1-\lambda) h ( y ) 
  \end{align*}
  
  \per
  \noindent
  Hence, $h|_{\fdom(h)}$ is order-concave.

 \item  
  \per
  Now, 
  assume that $\fdom(h) \neq \emptyset$.
  That is,
  there exists some $y \in \R^n$ with $h(y) = f(g^{(I)}(y)) \in \R$.
  Since $f$ is strict,
  we get $y' := g^{(I)}(y) \gg (\neginfty,\ldots,\neginfty)$.
  Let
  $I' : \{ 1,\ldots,m \} \to \{ \neginfty, \mathsf{id}, \infty \}$ 
  be defined by 
  \begin{align*}
    I'(k) 
    &=
    \begin{cases}
      \mathsf{id}      & \text{if } y'_{k\cdot} \in \R \\
      \mathsf{\infty} & \text{if } y'_{k\cdot} = \infty 
    \end{cases}
    && 
    \text{for all } k \in \{1,\ldots,m\}
    .
  \end{align*}
  
  \per
  \noindent
  Since $g$ is order-concave,
  we get
  $g^{(I)}_{k\cdot}(x) < \infty$ 
  for all 
  $x \in \R^n$
  and all 
  $k \in \{1,\ldots,m\}$ with $y'_{k\cdot} \in \R$.
  Since $f$ is order-concave,
  we get 
  $f^{(I')}(x) < \infty$ for all $x \in \R^n$.
  Thus, 
  by monotonicity, 
  we get 
  $f^{(I')} \circ g^{(I)}(x) = f^{(I')}(g^{(I)}(x)) < \infty$ for all $x \in \R^n$.
  Since we have 
  $h = (f \circ g)^{(I)} \leq f^{(I')} \circ g^{(I)}$ 
  by construction,
  we get 
  $h(x) < \infty$ for all $x \in \R^n$.
  \qed
 \end{enumerate}
\end{proof}

%% file: stratimp.tex
\np\section{Solving Systems of \maxmorcave Equations}
\label{s:strat:imp}


\per
\noindent
In this section,
we present our $\vee$-strategy improvement algorithm 
for computing least solutions of systems of \maxmorcave equations and prove its correctness.

\subsection{Systems of \maxmorcave Equations}

\per
\noindent
Assume that a fixed finite set $\vX$ of variables 
and a complete linearly ordered set $\D$ is given.
Assume that $\D$ is partially ordered by $\leq$.
We consider equations of the form 
$\vx = e$ over $\D$,
where $\vx \in \vX$ is a variable
and $e$ is an expression over $\D$.
A \emph{system} $\E$ of (fixpoint-)equations over $\D$ is a finite
set 
$ 
  \{ \vx_1 = e_1,\ldots,\vx_n = e_n \}
$ 
%
of equations, 
where
$\vx_1,\ldots,\vx_n$ are pairwise distinct variables.
We denote the set $\{\varx_1, \allowbreak \ldots, \allowbreak \varx_n\}$ of variables 
occurring in $\E$ by $\vX_\E$.
We drop the subscript, whenever it is clear from the context.


\per
For a variable assignment $\rho : \vX \to \D$,
an expression $e$ is mapped to a value 
$\sem{e}\rho$ 
by setting 
$	  \sem{\vx}\rho := \rho(\vx)$%
	  , and
$	  \sem{f(e_1,\ldots,e_k)}\rho := f(\sem{e_1}\rho,\ldots,\sem{e_k}\rho)$%
	  ,
where $\vx \in \vX$, $f$ is a $k$-ary operator ($k=0$ is possible; then $f$ is a constant), 
for instance $+$, and
$e_1,\ldots,e_k$ are expressions.
For every system $\E$ of equations,
we define the unary operator 
$\sem{\E}$ 
on 
$\vX \to \D$
by setting
$
  (\sem{\E}\rho)(\vx) 
  := \sem{e}\rho 
$
for all equations 
$\vx = e$ from $\E$
and all 
$\rho : \vX \to \D$.
A \emph{solution} is a fixpoint of $\sem\E$,
i.e.,
it is a variable assignment $\rho$
such that $\rho = \sem{\E}\rho$.
We denote the set of all solutions of $\E$ by $\Sol(\E)$.

\per
The set $\vX \to \D$ of all \emph{variable assignments} 
is a complete lattice.
For $\rho, \rho' : \vX \to \D$,
we write $\rho \ll \rho'$ (resp.\ $\rho \gg \rho'$) 
if and only if 
$\rho(\vx) < \rho'(\vx)$ (resp.\ $\rho(\vx) > \rho'(\vx)$) for all $\vx \in \vX$.
For $d \in \D$,
$\underline d$ denotes the variable assignment 
$\{ \vx \mapsto d \mid \vx \in \vX \}$.
A variable assignment $\rho$ with $\botvar \ll \rho \ll \topvar$
is called \emph{finite}.
A pre-solution (resp.\ post-solution) is a variable assignment $\rho$
such that  
$\rho \leq \sem{\E}\rho$ (resp.\ $\rho \geq \sem{\E}\rho$) holds.
The set of pre-solutions 
(resp.\ the set of post-solutions) 
is denoted 
by $\PreSol(\E)$ (resp.\ $\PostSol(\E)$).
The least solution (resp.\ the greatest solution)
of a system $\E$ of equations is denoted by
$\mu\sem\E$ (resp.\ $\nu\sem\E$),
provided that it exists.
For a pre-solution $\rho$ (resp.\ for a post-solution $\rho$),
$\mu_{\geq \rho}\sem\E$ (resp.\ $\nu_{\leq \rho}\sem\E$)
denotes the least solution that is greater than or equal to $\rho$
(resp.\ the greatest solution that is less than or equal to $\rho$).

\per
An expression $e$ (resp.\ an (fixpoint-)equation $\vx = e$ is called 
\emph{monotone} if and only if $\sem e$ is monotone.
In our setting, 
the fixpoint theorem of Knaster/Tarski can be stated as follows:
%
  every system $\E$ of monotone fixpoint equations over a complete lattice 
  has a least solution $\mu\sem\E$ and a greatest solution $\nu\sem\E$.
  Furthermore,
  we have
  $\mu\sem\E = \bigwedge \PostSol(\E)$
  and
  $\nu\sem\E = \bigvee \PreSol(\E)$.

\begin{definition}[\maxmorcave Equations]
  \per
  An expression $e$ (resp.\ fixpoint equation $\vx = e$) over $\CR$
  is called 
  \emph{\morcave} (resp.\ \cmorcave, resp.\ \mcave, resp.\ \cmcave)
  if and only if
  $\sem{e}$ is \morcave  (resp.\ \cmorcave, resp.\ \mcave, resp.\ \cmcave).
  An expression $e$ (resp.\ fixpoint equation $\vx = e$) over $\CR$
  is called \emph{\maxmorcave} (resp.\ \maxcmorcave, resp.\ \mcave, resp.\ \cmcave)
  if and only if 
  $e = e_1 \vee\cdots\vee e_k$,
  where $e_1,\ldots,e_k$ are
  \morcave (resp.\ \cmorcave, resp.\ \mcave, resp.\ \cmcave).
  \qed
\end{definition}


\per
\begin{example}
  \label{ex:wurzel:2}
  The square root operator $\sqrt\cdot : \CR\to\CR$ 
  (defined by $\sqrt x := \sup \; \{ y \in \R \mid y^2 \leq x \}$ for all $x \in \CR$)
  is \cmcave.
  The least solution of the system
  $\E = \{ \vx = \frac12 \vee \sqrt \vx \}$
  of \maxcmcave equations is
  $\mu\sem\E = 1$.
  \qed
\end{example}

\begin{definition}[$\vee$-strategies]
\per
A \emph{$\vee$-strategy $\sigma$} 
for a system $\E$ of equations
is a function that maps every expression 
$e_1 \vee \cdots \vee e_k$ 
occurring in $\E$
to one of the immediate sub-expressions $e_j$, $j \in \{1,\ldots,k\}$.
We denote the set of all 
$\vee$-strategies 
for $\E$ by 
$\MaxStrat_\E$.
We drop the subscript, whenever it is clear from the context.
The application $\E(\sigma)$ of $\sigma$ to $\E$ is defined by
$\E(\sigma) := \{ \vx = \sigma(e) \mid \vx = e \in \E \}$.
\end{definition}

\begin{example}
  \per
  The two $\vee$-strategies $\sigma_1,\sigma_2$ for the 
  system $\E$ of \maxcmcave equations defined in
  Example \ref{ex:wurzel:2}
  lead to the systems
  $\E(\sigma_1) = \{\vx = \frac12 \}$ and
  $\E(\sigma_2) = \{\vx = \sqrt \vx \}$
  of \cmcave equations.
  \qed
\end{example}

\subsection{The Strategy Improvement Algorithm}

\per
\noindent
We now present the $\vee$-strategy improvement algorithm 
in a general setting.
That is,
we consider arbitrary systems of monotone equations over arbitrary \emph{complete linearly ordered sets $\D$}.
The algorithm iterates over $\vee$-strategies.
It maintains a current $\vee$-strategy $\sigma$
and a current \emph{approximate} $\rho$ to the least solution.
A so-called \emph{$\vee$-strategy improvement operator} 
is used to determine a next, improved $\vee$-strategy $\sigma'$.
Whether or not a $\vee$-strategy $\sigma'$ is an \emph{improvement} 
of the current $\vee$-strategy $\sigma$ may depend on the current approximate $\rho$:

\begin{definition}[Improvements]
\per
	\label{d:alg:verbesserung}
	Let $\E$ be a system of monotone equations
	over a complete linearly ordered set.
	Let $\sigma, \sigma' \in \MaxStrat$ be
	$\vee$-strategies for $\E$ 
	and $\rho$ be a pre-solution of $\E(\sigma)$.
	The $\vee$-strategy $\sigma'$ is called an
	\emph{improvement of $\sigma$ w.r.t.\ $\rho$}
	if and only if 
	the following conditions are fulfilled:
    \begin{enumerate}
      \item
			If $\rho \notin \Sol(\E)$,
			then $\sem{\E(\sigma')}\rho > \rho$.
      \item
			For all expressions 
			$e = e_1 \vee \cdots \vee e_k$ of $\E$
			the following holds:
			If $\sigma'(e) \neq \sigma(e)$,
			then $\sem{\sigma'(e)} \rho > \sem{\sigma(e)} \rho$.
    \end{enumerate}
    
    \noindent
	A function $\Pv$ that assigns 
	an improvement of $\sigma$ w.r.t.\ $\rho$
	to every pair $(\sigma,\rho)$,
	where $\sigma$ is a $\vee$-strategy 
	and $\rho$ is a pre-solution of $\E(\sigma)$,
	is called a
	\emph{$\vee$-strategy improvement operator}.
	If it is impossible to improve $\sigma$ w.r.t.\ $\rho$,
	then we necessarily have $\Pv(\sigma,\rho) = \sigma$.
	\qed
\end{definition}

\begin{example}
  \per
  \label{ex:einfach:reicht:nicht}
  Consider the system
  $ 
    \E = \{ \vx_1 = \vx_2 + 1 \wedge 0, \vx_2 = -1 \vee \sqrt \vx_1 \}
  $ 
  %
  of \maxcmcave equations.
  Let $\sigma_1$ and $\sigma_2$ be the $\vee$-strategies for $\E$
  such that
  \begin{align*}
    \E(\sigma_1) &= \{ \vx_1 = \vx_2 + 1 \wedge 0, \vx_2 = -1 \},
    \text{ and}
    \\
    \E(\sigma_2) &= \{ \vx_1 = \vx_2 + 1 \wedge 0, \vx_2 = \sqrt \vx_1 \}
    .
  \end{align*}
  
  \noindent
  The variable assignment $\rho := \{ \vx_1 \mapsto 0, \vx_2 \mapsto -1 \}$ 
  is a solution and thus also a pre-solution of $\E(\sigma_1)$.
  The $\vee$-strategy $\sigma_2$ is an improvement of the 
  $\vee$-strategy $\sigma_1$
  w.r.t.\ $\rho$.
  \qed
\end{example}

\per
\noindent
We can now formulate the 
$\vee$-strategy improvement algorithm for
computing least solutions of systems of monotone equations
over complete linearly ordered sets. 
This algorithm is parameterized with a 
$\vee$-strategy improvement operator $\Pv$.
The input is a system $\E$ of monotone equations
over a complete linearly ordered set,
a $\vee$-strategy $\sigma_{\mathrm{init}}$ for $\E$, 
and
a pre-solution $\rho_{\mathrm{init}}$ 
of $\E(\sigma_{\mathrm{init}})$.
In order to compute the \emph{least} and not just some solution,
we additionally require that $\rho_{\mathrm{init}} \leq \mu\sem\E$ holds:

\vspace*{-2mm}

\per
\begin{algorithm}[H]
	$
	\\[0mm]
	\begin{array}{@{}l@{\text{:}\,}l@{}}
		\text{Parameter} 
		&
		\text{A $\vee$-strategy improvement operator $\Pv$} 
		   \\[1mm]
		\text{Input} 
		&
		\left\{
		\begin{array}{@{}l@{}}
			\text{\!\!-A system $\E$ of monotone equations
			           over a complete linearly ordered set } \\
			\text{\!\!-A $\vee$-strategy $\sigma_{\mathrm{init}}$ for $\E$} \\
			\text{\!\!-A pre-solution $\rho_{\mathrm{init}}$ 
			           of $\E(\sigma_{\mathrm{init}})$ 
			           with $\rho_{\mathrm{init}} \leq \mu\sem\E$} \\
		\end{array}
		\right.
		\\[5mm]
		\text{Output} 
		&
		\text{The least solution $\mu\sem\E$ of $\E$} 
	\end{array} \\[2mm]
    \sigma \GETS \sigma_{\mathrm{init}} ; \\
    \rho \GETS \rho_{\mathrm{init}} ; \\
    \VS
	\WHILE (\rho \notin \Sol(\E)) \;
	\{ \\
		\hspace*{0.5cm} \sigma \GETS \Pv(\sigma,\rho) ; \\ \hspace*{0.5cm}
		\rho \GETS \mu_{\geq \rho} \sem{\E(\sigma)} ; \\
	\} \\
	\VS
	\RETURN \rho;
	\\[-3mm]
	$
	\caption{The $\vee$-Strategy Improvement Algorithm}
	\label{alg:alg:stratimp}
\end{algorithm}

\vspace*{-2mm}

\begin{example}
  We consider the system 
  \begin{align}
    \textstyle
    \E = 
    \left\{ 
      \vx 
      = 
      \neginfty \vee \frac 1 2 \vee \sqrt \vx \vee \frac 7 8 + \sqrt{\vx - \frac{47}{64}}  
    \right\}
  \end{align}
  
  \noindent
  of \maxcmorcave equations.
  We start with the $\vee$-strategy $\sigma_0$ that leads to the system
  \begin{align}
    \E(\sigma_0) = \{ \vx = \neginfty \}
  \end{align} 
  
  \noindent
  of \cmorcave equations.
  Then 
  $\rho_0 := \neginftyvar$ 
  is a feasible solution of 
  $\E(\sigma_0)$.
  Since $\rho_0 \notin \Sol(\E)$,
  we improve $\sigma_0$ w.r.t.\ $\rho_0$ to the $\vee$-strategy 
  $\sigma_1$
  that gives us
  \begin{align}
    \E(\sigma_1) = \left\{ \vx = \frac 1 2 \right\}
    .
  \end{align}
  
  \noindent
  Then,
  $\rho_1 := \mu_{\geq \rho_0} \sem{\sigma_1} = \{ \vx \mapsto \frac12 \}$.
  Since $\sqrt{\frac 1 2} > \frac 1 2$
  and $\frac78 + \sqrt{\frac 1 2 - \frac{47}{64}} < \frac 1 2$
  hold,
  we improve the strategy $\sigma_1$ w.r.t.\ $\rho_1$
  to the $\vee$-strategy 
  $\sigma_2$ with
  \begin{align*}
    \E(\sigma_2) = \{ \vx = \sqrt \vx \}
    .
  \end{align*}
  
  \noindent
  We get $\rho_2 := \mu_{\geq \rho_1} \sem{ \sigma_2 } = \{ \vx \mapsto 1\}$.
  Since 
  $
  \frac78 + \sqrt{1 - \frac{47}{64}} 
  > 
  \frac78 + \sqrt{1 - \frac{60}{64}} 
  =
  \frac98
  > 1
  $,
  we get 
  $\sigma_3 = \{ \vx = \frac78 + \sqrt{\vx - \frac{47}{64}} \}$.
  Finally we get 
  $
    \rho_3 
    := 
    \mu_{\geq \rho_2} \sem{\sigma_3}
    =
    \{ \vx \mapsto 2 \}
  $.
  The algorithm terminates,
  because $\rho_3$ solves $\E$.
  Therefore, $\rho_3 = \mu\sem\E$.
  \qed
\end{example}

\per
\noindent
In the following lemma, we collect basic properties that can be proven by induction straightforwardly:

\begin{lemma}
	\label{l:alg:sequence}
	\per
	Let $\E$ be a system of monotone equations
	over a complete linearly ordered set.
    For all $i \in \N$, 
    let
    $\rho_i$ be the value of the program variable $\rho$ 
    and
    $\sigma_i$ be the value of the program variable $\sigma$ 
    in the $\vee$-strategy improvement algorithm  
    (Algorithm \ref{alg:alg:stratimp})
    after the $i$-th evaluation of the loop-body.
	The following statements hold for all $i \in \N$:

	\begin{enumerate}
	  \item
			$\rho_i \leq \mu\sem\E$.
      \item
			$\rho_i \in \PreSol(\E(\sigma_{i+1}))$.
      \item
			If $\rho_i < \mu\sem\E$, then $\rho_{i+1} > \rho_i$.
      \item
			If $\rho_i = \mu\sem\E$, then $\rho_{i+1} = \rho_i$.
	\end{enumerate}
	
  \noindent
  If the execution of the $\vee$-strategy improvement algorithm terminates, 
  then the least solution $\mu\sem\E$ of $\E$ is computed.
  \qed
\end{lemma}

\per
\noindent
In the following, we apply our algorithm to solve  
systems of \maxmorcave equations.
In the next subsection,
we show that our algorithm terminates in this case.
More precisely,
it returns the least solution at the latest after considering every $\vee$-strategy at most $\abs \vX$ times.
We additionally provide an important characterization of 
$\mu_{\geq\rho}\sem{\E(\sigma)}$
which allows us to compute it using convex optimization techniques.
Here, $\sigma$ are the $\vee$-strategies and $\rho$ are the pre-solutions $\rho$ of $\E(\sigma)$ 
that can be encountered during the execution of the algorithm.

\subsection{Feasibility}

\per
\noindent
In this subsection,
we extend the notion of feasibility as defined in 
Definition~\ref{def:feas}.
We then show that feasibility is preserved during the execution of the $\vee$-strategy improvement algorithm.
In the next subsection,
we finally make use of the feasibility.

We denote by $\E[x_1/\vX_1 , \ldots, x_n/\vX_n]$
the equation system that is obtained from the equation system $\E$ by simultaneously replacing,
for all $i \in \{1,\ldots,n\}$,
every occurrence of a variable from the set $\vX_i$ in the right-hand sides of $\E$ by the value $x_i$.

\begin{definition}[Feasibility]
\label{d:feasibility}
\per
  Let $\E$ be a system of \morcave equations. 
  A finite solution $\rho$ of $\E$ is called \emph{($\E$-)feasible}
  if and only if
  $\rho$ is a feasible fixpoint of $\sem\E$.
  A pre-solution $\rho$ of $\E$ 
  with $\sem\E\rho \gg \neginftyvar$
  is called \emph{($\E$-)feasible}
  if and only if
  $\rho'|_{\vX'}$ is a feasible finite solution of 
  $\E' := \{ \vx = e \in \E \mid \vx \in \vX' \}[\infty / (\vX\setminus\vX')] $, 
  where 
  $\rho' := \mu_{\geq \rho} \sem\E$ 
  and 
  $\vX' := \{ \vx \in \vX \mid \rho'(\vx) < \infty \}$.
  %
  A pre-solution $\rho$ of $\E$ is called \emph{feasible} 
  if and only if 
  $e = \neginfty$ for all $\vx = e \in \E$ with $\sem e \rho = \neginfty$,
  and 
  $\rho |_{\vX'}$ is a feasible pre-solution of 
  $\E' := \{ \vx = e \in \E \mid \vx \in \vX' \} [\neginfty / (\vX\setminus\vX')]$,
  where $\vX' := \{ \vx \mid \vx = e \in \E, \sem e \rho > \neginfty \}$.
%
  %
  \qed
\end{definition}

\begin{example}
\per
  We consider the system 
  $\E = \{ \vx = \sqrt \vx \}$
  of \mcave equations.
  For all $x \in \CR$,
  let $\underline x := \{ \vx \mapsto x \}$.
  From Example \ref{ex:feas:function:1},
  we know that
  the solution $\underline 0$ is not feasible, 
  whereas the solution
  $\underline 1$ is feasible.
  Thus,
  $\underline{x}$ is a feasible pre-solution for all $x \in (0,1]$.
  Note that $\underline 1$ is the only feasible finite solution of $\E$ and
  thus, 
  by Lemma \ref{l:feasible:is:greatest}, 
  the greatest finite pre-solution of $\E$.
  \qed
\end{example}

\begin{example}
\label{ex:strategie:ok}
\per
  Let us consider the system 
  $
    \E = \{ \vx_1 = \vx_2 + 1 \wedge 0, \vx_2 = \sqrt \vx_1 \} 
  $
  of \mcave equations.
  From Example \ref{ex:two:dim:feasible:function}
  it follows that
  $\rho := \{ \vx_1 \mapsto 0, \vx_2 \mapsto 0 \}$ 
  is a feasible finite fixpoint of $\sem\E$.
  Thus, 
  $\{ \vx_1 \mapsto 0, \vx_2 \mapsto x \}$ 
  is a feasible pre-solution for all $x \in [-1,0]$.
  The solution 
  $\{ \vx_1 \mapsto \neginfty, \vx_2 \mapsto \neginfty \}$
  is not feasible,
  since the right-hand sides evaluate to $\neginfty$,
  although they are not $\neginfty$.
  \qed
\end{example}

\per
\noindent
The following two lemma 
imply that our $\vee$-strategy improvement algorithm
stays in the feasible area, whenever it is started in the feasible area.

\begin{lemma}
  \label{l:gro:erh:zul}
  \per
  Let $\E$ be a system of \morcave equations and
  $\rho$ be a feasible pre-solution of $\E$.
  Every pre-solution $\rho'$ of $\E$ with 
  $\rho \leq \rho' \leq \mu_{\geq \rho}\sem\E$ 
  is feasible. 
\end{lemma}

\begin{proof}
  The statement is  an immediate consequence of the definition.
  \qed
\end{proof}


\begin{lemma}
	\label{l:rat:stratimperhaeltzulaessigkeit}
    \per
	Let $\E$ be a system of \maxmorcave equations, 
	$\sigma$ be a $\vee$-strategy for $\E$, 
	$\rho$
	be a feasible 
    solution 
	of $\E(\sigma)$,
	and
	$\sigma'$ be an improvement of $\sigma$ w.r.t.\ $\rho$.
	Then $\rho$ is a feasible pre-solution of $\E(\sigma')$.
%
\end{lemma}

\begin{proof}
  \newcommand{\vXimp}{\vX^{\mathsf{imp}}}
  \newcommand{\Eold}{\E^{\mathsf{old}}}
  \newcommand{\vXold}{\vX^\mathsf{old}}
  \per
  Let $\rho^* := \mu_{\geq \rho}\sem{\E(\sigma')}$.
  We w.l.o.g.\ assume that 
  $\neginftyvar \ll \rho^* \ll \inftyvar$.
  Hence, 
  $\rho \ll \inftyvar$.
  Let 
  \begin{align*}
    \vXold 
      &:= \{ \vx \in \vX \mid \rho(\vx) > \neginfty \}, \text{and } \\
    \Eold 
      &:= \{ \vx = e \in \E(\sigma) \mid \vx \in \vXold \} 
        [\neginfty / (\vX \setminus \vXold)]
     .
  \end{align*}
  
  \per
  \noindent
  Hence, $\rho|_{\vXold}$ is a feasible finite solution of $\Eold$,
  i.e., a feasible finite fixpoint of $\sem\Eold$.
  Therefore,
  there exist
  $\vX_1 \dcup\cdots\dcup \vX_k = \vXold$
  with 
  \begin{align}
    \vX_1 
    \stackrel{\sem\Eold,\rho|_{\vXold}}\rightarrow
    \cdots
    \stackrel{\sem\Eold,\rho|_{\vXold}}\rightarrow 
    \vX_k
  \end{align}
  
  \per
  \noindent
  such that, for each $j \in \{1,\ldots,k\}$,
  there exists some pre-fixpoint $\rho'$ 
  of $\sem{\Eold} \leftarrow \rho|_{\vXold \setminus \vX_j}$
  with $\rho' \ll \rho|_{\vX_j}$
  such that 
  $
    \mu_{\geq\rho'}(\sem{\Eold} \leftarrow \rho|_{\vXold \setminus \vX_j})
    =
    \rho|_{\vX_j}
  $.

  \per
  Let
  $\vXimp := \{ \vx \in \vX \mid \rho^*(\vx) > \rho(\vx) \}$, 
  $\vX_j' := \vX_j \setminus \vXimp$
  for all 
  $j \in \{1,\ldots,k\}$, and
  $\vX_{k+1}' := \vXimp$.
  Obviously,
  we have
  $\vX_1' \dcup \cdots 
  \dcup \vX_{k+1}' = \vX$.
  It remains to show that the following properties are fulfilled:
  \begin{enumerate}
    \item
      \per
	  $
	    \vX_1' 
	    \stackrel{\sem{\E(\sigma')},\rho^*}\rightarrow
	    \cdots
	    \stackrel{\sem{\E(\sigma')},\rho^*}\rightarrow 
	    \vX_{k+1}'
	  $
	\item
      \per
	  For each $j \in \{1,\ldots,k+1\}$,
	  there exists some pre-fixpoint 
	  $\rho'$ with $\rho' \ll \rho^* |_{\vX_j'}$
	  such that 
	  $
	    \mu_{\geq\rho'}(\sem{\E(\sigma')} \leftarrow \rho^*|_{\vX\setminus \vX_j'}) 
	    =
	    \rho^*|_{\vX_j'}
	  $.
  \end{enumerate}

  \per
  \noindent
  In order to prove statement 1,
  let $j \in \{1,\ldots,k\}$.
  We have to show that 
  \begin{align*}
    \vX_1' \dcup\cdots\dcup \vX_j' 
    \stackrel{\sem{\E(\sigma')},\rho^*}\rightarrow
    \vX_{j+1}' \dcup\cdots\dcup \vX_{k+1}' 
    .
  \end{align*}
  
  \per
  \noindent
      Since 
      $
        \vX_1 \dcup\cdots\dcup \vX_j
        \stackrel{\sem{\Eold},\rho|_{\vXold}}\rightarrow
        \vX_{j+1} \dcup\cdots\dcup \vX_{k}
      $,
      there exists some variable assignment 
      $\rho' : \vX_{j+1} \dcup\cdots\dcup \vX_{k} \to \R$
      with
      $\rho' \ll \rho|_{\vX_{j+1} \dcup\cdots\dcup \vX_{k}}$ 
      such that 
      \begin{align}
       \label{eq:l:impr:1}
        (\sem{\Eold} (\rho|_{\vXold} \oplus \rho'))|_{\vX_1 \dcup\cdots\dcup \vX_j}
        =
        (\sem{\Eold} (\rho|_{\vXold}))|_{\vX_1 \dcup\cdots\dcup \vX_j}
        .
      \end{align}
      
      \per
      \noindent
      We define $\rho'' : \vX_{j+1}' \dcup\cdots\dcup \vX_{k+1}' \to \R$ by
      \begin{align*}
        \rho'' (\vx) 
        {=}
        \begin{cases}
          \rho'(\vx) &\text{if } \vx \in \vX_{j+1}' \dcup\cdots\dcup \vX_{k}' \\
          \rho(\vx) &\text{if } \vx \in \vX_{k+1}' \text{ and } \vx \in \vXold \\
          \rho^*(\vx) - 1 &\!\text{if } \vx \in \vX_{k+1}' \text{ and } \vx\notin\vXold \\
        \end{cases}
        && 
        \text{for all } \vx \in \vX_{j+1}' \dcup\cdots\dcup \vX_{k+1}'
        .
      \end{align*}
  
  \ok
  \noindent
  By construction,
  we have
  $
    \rho'' 
    \ll 
    \rho^*|_{\vX_{j+1}' \dcup\cdots\dcup \vX_{k+1}'}
  $.
  Hence,
  we get 
  \begin{align*}
      (\sem{\E(\sigma')} (\rho^*))|_{\vX_1' \dcup\cdots\dcup \vX_j'}
        &\geq
        (\sem{\E(\sigma')} (\rho^* \oplus \rho''))|_{\vX_1' \dcup\cdots\dcup \vX_j'}
        & \text{($\rho^* \geq \rho^* \oplus \rho''$)}
      \\&\geq
        (\sem{\Eold} (\rho|_{\vXold} \oplus \rho'))|_{\vX_1' \dcup\cdots\dcup \vX_j'}
      \\&=
        (\sem{\Eold} (\rho|_{\vXold}))|_{\vX_1' \dcup\cdots\dcup \vX_j'}
        & \text{(because of \eqref{eq:l:impr:1})}
      \\&=
        (\sem{\Eold} (\rho^*|_{\vXold}))|_{\vX_1' \dcup\cdots\dcup \vX_j'}
        & \!\!\!\!\text{(because of Lemma \ref{l:partition:nach:oben:unabhae})}
      \\&=
        (\sem{\E(\sigma')} (\rho^*))|_{\vX_1' \dcup\cdots\dcup \vX_j'}
  .
  \end{align*}
  
  \noindent
  Thus,
  $
     (\sem{\E(\sigma')} (\rho^* \oplus \rho''))|_{\vX_1' \dcup\cdots\dcup \vX_j'}
     =
     (\sem{\E(\sigma')} (\rho^*))|_{\vX_1' \dcup\cdots\dcup \vX_j'}
  $.
  This proves statement 1.
  
  \per
  In order to prove statement 2,
  let $j \in \{1,\ldots,k+1\}$.
  We distinguish 2 cases.
  Firstly, 
  assume that $j \leq k$.
  Since $\rho|_{\vXold}$ is a feasible finite fixpoint of $\sem{\Eold}$,
  there exists some pre-fixpoint
  $\rho'$ with $\rho' \ll \rho|_{\vX_j} = \rho^* |_{\vX_j}$
  such that 
  $
    \mu_{\geq\rho'} (\sem{\Eold} \leftarrow \rho|_{\vXold\setminus \vX_j})
    =
    \rho|_{\vX_j}
    = 
    \rho^* |_{\vX_j}
  $.
  Using monotonicity,
  we get 
  $
    \mu_{\geq\rho'} (\sem{\Eold} \leftarrow \rho^*|_{\vXold\setminus \vX_j})
    \allowbreak
    =
    \rho|_{\vX_j}
    = 
    \rho^* |_{\vX_j}
  $.
  Hence,
  $\rho'|_{\vX_j'} : \vX_j' \to \R$, 
  $\rho'|_{\vX_j'} \ll \rho|_{\vX_j'} = \rho^* |_{\vX_j'}$,
  and  
  $
    \mu_{\geq\rho'|_{\vX_j'}} (\sem{\Eold} \allowbreak \leftarrow \rho^*|_{\vXold\setminus \vX_j'}) 
    =
    \mu_{\geq\rho'|_{\vX_j'}} 
      (\sem{\E(\sigma')} \leftarrow \rho^*|_{\vX\setminus \vX_j'}) 
    =
    \rho^*|_{\vX_j'}
  $.
  This proves statement 2 for $j \leq k$.
  Now, assume that $j = k+1$.
  By definition of $\vX_{k+1}'$,
  $\rho|_{\vX_{k+1}'} \ll \rho^*|_{\vX_{k+1}'}$.
  Moreover, 
  we get immediately
  that 
  $\rho|_{\vX_{k+1}'}$
  is a pre-fixpoint of
  $\sem{\E(\sigma')} \leftarrow \rho^*|_{\vX\setminus \vX_{k+1}'}$ 
  and 
  $
    \mu_{\geq \rho|_{\vX_{k+1}'}}
      (\sem{\E(\sigma')} \leftarrow \rho^*|_{\vX\setminus \vX_{k+1}'}) 
    =
    \rho^*|_{\vX_{k+1}'}
  $.
  This proves statement 2.
  \qed
\end{proof}

\begin{example}
  \per
  We continue Example \ref{ex:einfach:reicht:nicht}.
  Obviously,
  $\rho = \{ \vx_1 \mapsto 0, \vx_2 \mapsto -1 \}$ is a feasible solution of 
  $\E(\sigma_1) =  \{ \vx_1 = \vx_2 + 1 \wedge 0, \vx_2 = -1 \}$.
  The $\vee$-strategy $\sigma_2$ is 
  an improvement of the $\vee$-strategy $\sigma_1$
  w.r.t.\ $\rho$.
  By lemma \ref{l:rat:stratimperhaeltzulaessigkeit},
  $\rho$ is also a feasible pre-solution of 
  $\E(\sigma_2) = \{ \vx_1 = \vx_2 + 1 \wedge 0, \vx_2 = \sqrt \vx_1 \}$.
  The fact that 
  $\rho$ is a feasible pre-solution of $\E(\sigma_2)$ is also shown in 
  Example \ref{ex:strategie:ok}.
  \qed
\end{example}

\per
\noindent
The above two lemmas ensure that our $\vee$-strategy improvement 
algorithm stays in the feasible area,
whenever it is started in the feasible area.
In order to start in the feasible area,
we in the following simply assume w.l.o.g.\ that
each equation of $\E$ is of the form $\vx = \neginfty \vee e$.
We say that such a system of fixpoint equations is in \emph{standard form}.
Then, we start our $\vee$-strategy improvement algorithm with 
a $\vee$-strategy 
$\sigma_\mathrm{init}$
such that 
$\E(\sigma_\mathrm{init}) = \{ \vx = \neginfty \mid \vx \in \vX \}$.
In consequence, $\neginftyvar$ is a feasible solution of
$\E(\sigma_\mathrm{init})$.
We get:

\begin{lemma}
	\label{l:alg:sequence:feasible}
	\per
	Let $\E$ be a system of \maxmorcave equations.
    For all $i \in \N$, 
    let
    $\rho_i$ be the value of the program variable $\rho$ 
    and
    $\sigma_i$ be the value of the program variable $\sigma$ 
    in the $\vee$-strategy improvement algorithm  
    (Algorithm \ref{alg:alg:stratimp})
    after the $i$-th evaluation of the loop-body.
	Then, $\rho_i$ is a feasible pre-solution of $\E(\sigma_{i+1})$ for all $i \in \N$.
  \qed
\end{lemma}

\begin{example}
  \per
  We again consider the system
  $
    \E 
    = 
    \{ 
      \vx_1 = \neginfty \vee \vx_2 + 1 \wedge 0, \;
      \vx_2 = \neginfty \vee -1 \vee \sqrt \vx_1
    \}
  $
  of \maxmorcave equations
  introduced in Example \ref{ex:einfach:reicht:nicht}.
  A run of our $\vee$-strategy improvement algorithm gives us
  \begin{align*}
    \E(\sigma_0) &= \{ \vx_1 = \neginfty ,\; \vx_2 = \neginfty \} &
    \rho_0 &= \{ \vx_1 \mapsto \neginfty ,\; \vx_2 \mapsto \neginfty \} \\ 
    \E(\sigma_1) &= \{ \vx_1 = \neginfty ,\; \vx_2 = -1 \} &
    \rho_1 &= \{ \vx_1 \mapsto \neginfty ,\; \vx_2 \mapsto -1 \} \\ 
    \E(\sigma_2) &= \{ \vx_1 = \vx_2 + 1 \wedge 0 ,\; \vx_2 = -1 \} &
    \rho_2 &= \{ \vx_1 \mapsto 0 ,\; \vx_2 \mapsto -1 \} \\ 
    \E(\sigma_3) &= \{ \vx_1 = \vx_2 + 1 \wedge 0 ,\; \vx_2 = \sqrt {\vx_1} \} &
    \rho_3 &= \{ \vx_1 \mapsto 0 ,\; \vx_2 \mapsto 0 \} 
  \end{align*}
  
  \noindent
  By Lemma \ref{l:alg:sequence:feasible},
  $\rho_i$ is a feasible pre-solution of $\E(\sigma_{i+1})$
  for all $i = \{0,1,2\}$.
  \qed
\end{example}

\subsection{Evaluating $\vee$-Strategies / Solving Systems of \Morcave Equations}


\per
\noindent
It remains to develop a method for computing 
$\mu_{\geq \rho}\sem{\E}$
under the assumption that
$\rho$ is a feasible pre-solution of the 
system $\E$ of \morcave equations.
This is an important step in our $\vee$-strategy improvement algorithm 
(Algorithm \ref{alg:alg:stratimp}).
Before doing this, 
we introduce the following notation for the sake of simplicity:

\begin{definition}
  \label{d:supsol}
  \per
  Let $\E$ be a system of \morcave equations
  and $\rho$ a pre-solution of $\E$.
  Let 
  \begin{align}
    \label{eq:l:eindeutige:loesung:1}
    \vX^\neginfty_\rho &:= \{ \vx \mid \vx = e \in \E ,\; \sem{e} \rho = \neginfty \} \\
    \label{eq:l:eindeutige:loesung:2}
    \vX^\infty_\rho       &:= \{ \vx \mid \vx = e \in \E ,\; \sem{e} \rho = \infty \} \\
    \label{eq:l:eindeutige:loesung:3}
    \vX'_\rho               &:= \vX \setminus (\vX^\neginfty_\rho \cup \vX^\infty_\rho) 
                           = \{ \vx \mid \vx = e \in \E ,\; \sem{e} \rho \in \R \}
    \\
    \label{eq:l:eindeutige:loesung:4}
    \E'_\rho             &= \{ \vx = e \in \E \mid \vx \in \vX'_\rho \} [\neginfty/\vX^\neginfty_\rho, \infty/\vX^\infty_\rho]
  \end{align}
    
  \per
  \noindent
  The pre-solution $\suppresol_\rho\sem\E$ of $\E$ is defined by
  \begin{align}
    \label{eq:l:eindeutige:loesung:5}
    \suppresol_\rho\sem\E(\vx)
    &:=
    \begin{cases} 
      \neginfty & \text{if } \vx \in \vX^\neginfty_\rho \\
      \sup\; \{ \hat\rho(\vx) \mid \hat\rho:\vX'_\rho\to\R ,\; \hat\rho \leq \sem{\E'}\hat\rho\} & \text{if }\vx\in\vX'_\rho \\
      \infty & \text{if }\vx\in\vX^\infty_\rho
    \end{cases}
  \end{align}

  \noindent
  for all $\vx \in \vX$.  
  \qed
\end{definition}

\begin{remark}
  \per
  The variables assignment $\suppresol_\rho\sem\E$ is by construction a pre-solution of $\E$,
  but, as we will see in Example \ref{ex:strict:ineq:rho:star},
  not necessarily a solution of $\E$.
  \qed
\end{remark}

\ok
\noindent
Under some constraints,
we can compute $\suppresol_\rho\sem\E$ by solving 
$\abs\vX$ convex optimisation problems of linear size.
This can be done by general convex optimization methods.
For further information on convex optimization, we refer, for instance,
to \citet{nemirovski05}.

\begin{lemma}
  \label{l:convex:opt}
  Let $\E$ be a system of \mcave equations and $\rho$ a pre-solution of $\E$.
  Then, the pre-solution $\suppresol_\rho\sem\E$ of $\E$ can be computed by solving at most 
  $\abs\vX$ convex optimization problems.
\end{lemma}

\begin{proof}
  Let $\vX^\neginfty_\rho$, $\vX^\infty_\rho$, $\vX'_\rho$, and $\E'_\rho$ be defined as in Definition \ref{d:supsol}.
  We have to compute 
  $
    \suppresol_\rho\sem\E(\vx)
    =
    \sup\; \{ \hat\rho(\vx) \mid \hat\rho:\vX'_\rho\to\R ,\; \hat\rho \leq \sem{\E'}\hat\rho\} 
    =
    \sup\; \{ \hat\rho(\vx) \mid \hat\rho:\vX'_\rho\to\R ,\; (\mathsf{id} - \sem{\E'}) \hat\rho \leq 0 \} 
  $
  for all 
  $\vx\in\vX'_\rho$.
  Here, $\mathsf{id}$ denotes the identity function.
  Therefore, 
  since $\mathsf{id}$ is affine, 
  $\sem{\E'}$ is concave 
  (considered as a function that maps values from $\vX'_\rho\to\R$ to values from $\vX'_\rho\to(\R \cup \{\neginfty\}$), 
  and thus $- \sem{\E'}\hat\rho$ is convex 
  (considered as a function that maps values from $\vX'_\rho\to\R$ to values from $\vX'_\rho\to(\R \cup \{\infty\}$),
  the mathematical optimization problem
  $
    \sup\; \{ \hat\rho(\vx) \mid \hat\rho:\vX'_\rho\to\R ,\; (\mathsf{id} - \sem{\E'})\hat\rho \leq 0 \} 
  $
  is a convex optimization problem.
  \qed
\end{proof}

\per
\noindent
We will use $\suppresol_\rho\sem\E$ iteratively to compute $\mu_{\geq\rho}\sem\E$ under the
assumption that $\rho$ is a feasible pre-solution of the system $\E$ of \morcave equations.
As a first step in this direction,
we prove the following lemma,
which gives us at least a method for computing $\mu_{\geq\rho}\sem\E$ 
under the assumption that $\E$ is a system of \emph\cmorcave equations.


\begin{lemma}
\label{l:eindeutige:loesung}
  \per
  Let $\E$ be a system of \morcave equations
  and $\rho$ a feasible pre-solution of $\E$.
  Let 
  $\vX^\neginfty_\rho$,
  $\vX^\infty_\rho$, and 
  $\vX'_\rho$
  be defined as in Definition \ref{d:supsol} (\eqref{eq:l:eindeutige:loesung:1} - \eqref{eq:l:eindeutige:loesung:3}).
  Then:
  \begin{align}
    \label{eq:l:eindeutige:loesung:5:strich}
    \mu_{\geq\rho}\sem\E (\vx) 
      &= \suppresol_\rho\sem\E(\vx) = \neginfty 
      && \text{for all } \vx \in \vX^\neginfty_\rho
    \\
    \label{eq:l:eindeutige:loesung:6:strich}
    \mu_{\geq\rho}\sem\E (\vx) 
      &\geq \suppresol_\rho\sem\E(\vx)
      && \text{for all }\vx\in\vX'_\rho
    \\
    \label{eq:l:eindeutige:loesung:7:strich}
    \mu_{\geq\rho}\sem\E (\vx)   
      &= \suppresol_\rho\sem\E(\vx) = \infty 
      && \text{for all }\vx\in\vX^\infty_\rho
  \end{align}
  
  
  \per
  \noindent
  If $\E$ is a system of \cmorcave equations,
  then the inequality in \eqref{eq:l:eindeutige:loesung:6:strich} is in fact an equality,
  i.e., we have
  \begin{align}
    \label{eq:l:eindeutige:loesung:8}
    \mu_{\geq\rho}\sem\E 
      &=
      \suppresol_\rho\sem\E
      .
  \end{align}
\end{lemma}

\begin{proof}
  \per
  Let $\E'_\rho$ be defined as in Definition \ref{d:supsol} \eqref{eq:l:eindeutige:loesung:4}.
  We first prove \eqref{eq:l:eindeutige:loesung:5:strich} - \eqref{eq:l:eindeutige:loesung:7:strich}.
  Let $x \in \vX$.
  If $\vx \in \vX^\neginfty_\rho \cup \vX^\infty_\rho$, 
  then the statement is obviously fulfilled,
  because $\rho$ is feasible and thus
  $e=\neginfty$ for all equations $\vx = e$ from $\E$ with $\sem e \rho = \neginfty$.
  This gives us \eqref{eq:l:eindeutige:loesung:5:strich} and \eqref{eq:l:eindeutige:loesung:7:strich}.
  Assume now that $\vx \in \vX'_\rho$.
  Let 
  $\rho' := \rho|_{\vX'_\rho}$ 
  and 
  $\rho^* := \mu_{\geq\rho'}\sem{\E'_\rho}$.
  We have to show that
  \begin{align}
    \rho^*(\vx)
    \geq
    \sup\;\{ \hat\rho(\vx) \mid \hat\rho:\vX'_\rho \to\R ,\; \hat\rho \leq \sem{\E'_\rho}\hat\rho \}
    . 
  \end{align}

     \per
     \noindent
     If $\rho^*(\vx) = \infty$, 
     there is nothing to prove.
      Therefore,
      assume that
      $\rho^*(\vx) < \infty$.
      Then $\rho^*(\vx) \in \R$.
      Let $\vX''_\rho := \{ \vx'' \in \vX'_\rho \mid \rho^*(\vx'') < \infty \}$.
      Then, $\vX''_\rho = \{ \vx'' \in \vX'_\rho \mid \rho^*(\vx'') \in \R \}$.
      Let $\E''_\rho := \{ \vx'' = e \in \E'_\rho \mid \vx'' \in \vX''_\rho \}[\infty/(\vX'_\rho\setminus\vX''_\rho)]$,
      and $\rho'' := \rho|_{\vX''_\rho}$.
      The pre-solution $\rho''$ of $\E''_\rho$ is feasible. 
      Hence, $\rho^*|_{\vX''_\rho}$ is a feasible finite pre-solution of $\E''_\rho$,
      i.e., a feasible finite fixpoint of $\sem{\E''_\rho}$.
      Therefore,
      we finally get \eqref{eq:l:eindeutige:loesung:6:strich}
      using Lemma \ref{l:feasible:is:greatest}.

      \per
      Before we actually prove \eqref{eq:l:eindeutige:loesung:8},
      we start with an easy observation.
      The sequence $(\sem{\E'_\rho}^k \rho')_{k\in\N}$ is increasing, 
      because $\rho'$ is a pre-solution of $\E'_\rho$.
      Further 
      $\sem{\E'_\rho}^k \rho' : \vX'_\rho \to \R$
      and 
      $\sem{\E'_\rho}^k \rho' \leq \sem{\E'_\rho} ( \sem{\E'_\rho}^k \rho' )$
      for all $k \in \N$.
      Hence, we get
      \begin{align}
        \label{eq:besser:als:kleene}
        \sup\;\{ 
          \hat\rho(\vx) 
          \mid 
          \hat\rho:\vX'_\rho\to\R ,\; \hat\rho \leq \sem{\E'_\rho}\hat\rho
        \} 
        \geq
        \sup\;\{ 
          (\sem{\E'_\rho}^k \rho')(\vx) \mid k \in \N
        \}
      \end{align}
      
  \per
  \noindent
  Now, assume that $\E$ is a system of \cmorcave equations.
  In order to prove \eqref{eq:l:eindeutige:loesung:8},
  it remains to show that 
  $ 
    \rho^*(\vx)
    \leq
    \sup\;\{ \hat\rho(\vx) \mid \hat\rho:\vX'_\rho\to\R ,\; \hat\rho \leq \sem{\E'_\rho}\hat\rho \}
  $. 
      Since $\sem{\E'_\rho}$
      is monotone and upward-chain-continuous on 
      $\{ \hat\rho : \vX'_\rho\to\CR \mid \hat\rho \geq \rho' \}$, 
      we have
      $
        \rho^*
        =
        \bigvee \{ \sem{\E'_\rho}^k \rho' \mid k \in \N \}
      $.
      Using \eqref{eq:besser:als:kleene}, this gives us
      $ 
        \rho^*(\vx)
        \leq
        \sup\;\{ 
          \hat\rho(\vx) 
          \mid 
          \hat\rho:\vX'_\rho \to\R ,\; \hat\rho \leq \sem{\E'_\rho}\hat\rho
        \} 
      $, 
      as desired.
%
%
%
%
  \qed
\end{proof}

\per
\noindent
If the equations are \morcave but \emph{not} \cmorcave, 
then the inequality in
\eqref{eq:l:eindeutige:loesung:6:strich}
can indeed be strict as
the following example shows.

\begin{example}
  \label{ex:strict:ineq:rho:star}
  \per
  Let us consider the following system $\E$ of \morcave equations:
  \begin{align}
    \label{eq:ex:strict:ineq:rho:star:1}
    \vx_1 &= 1 &
    \vx_2 &= \vx_1 + \vx_2 &
    \vx_3 &= 
      \begin{cases}
        0 &\text{if } \vx_2 < \infty \\
        1 &\text{if } \vx_2 = \infty
      \end{cases}
  \end{align}
  
  \per
  \noindent 
  Observe that the third equation is \emph{not} \cmorcave,
  since,
  for the ascending chain 
  $C = \{ \{ \vx_2 \mapsto k \} \mid k \in \N \}$,
  we have
  $\bigvee \{ \sem{e} \rho \mid \rho \in C \} = 0 < 1 = \sem{e} (\bigvee C)$,
  where $e$ denotes the right-hand side of the third equation.
  The variable assignment 
  \begin{align}
    \label{eq:ex:strict:ineq:rho:star:2}
    \rho := \{ \vx_1 \mapsto 0 ,\; \vx_2 \mapsto 0 ,\; \vx_3 \mapsto 0 \}
  \end{align}
  
  \noindent
  is a feasible pre-solution,
  since 
  \begin{align}
    \label{eq:ex:strict:ineq:rho:star:3}
    \rho^* 
    := 
    \mu_{\geq\rho} \sem{\E} 
    = 
    \{ \vx_1 \mapsto 1 ,\; \vx_2 \mapsto \infty ,\; \vx_3 \mapsto 1 \}
  \end{align}
  
  \per
  \noindent
  is a feasible solution of $\E$.
  Now, let the variable assignment $\rho_1$ be defined by
  \begin{align}
    \label{eq:ex:strict:ineq:rho:star:4}
    \rho_1 &:= \suppresol_\rho\sem\E
    .
  \end{align}
  
  \per
  \noindent
  Lemma \ref{l:eindeutige:loesung}
  gives us $\rho_1 \leq \rho^*$,
  but not $\rho_1 = \rho^*$.
  Indeed,
  we have
  \begin{align}
    \label{eq:ex:strict:ineq:rho:star:5}
    \rho_1 
    = 
    \{ \vx_1 \mapsto 1 ,\; \vx_2 \mapsto \infty ,\; \vx_3 \mapsto 0 \}
    <
    \rho^*
    .
  \end{align}
  \per
  \noindent
  We emphasize that $\rho_1(\vx_3) = 0$,
  because $\sem{e}\hat\rho = 0$ for all $\hat\rho : \vX \to \R$,
  where $e$ denotes the right-hand side of the third equation of \ref{eq:ex:strict:ineq:rho:star:1}.

  \per
  How we can actually compute $\rho^*$,
  remains an open question. 
  The discontinuity at $\vx_2 = \infty$
  is the reason for the strict inequality in 
  \eqref{eq:ex:strict:ineq:rho:star:5}.
  However, since upward discontinuities can only be present at $\infty$,
  there are at most $n$ upward discontinuities,
  where $n$ is the number of variables of the equation system.
  Hence, we could think of using \eqref{eq:ex:strict:ineq:rho:star:4}
  to get over at least one discontinuity. 
  \per
  Let us perform a second iteration for the example.
  We know that $\rho_1 \leq \rho^*$.
  Moreover, by definition, $\rho_1$ is also a feasible pre-solution of $\E$.
  For the variable assignment $\rho_2$ that is defined by
  $
      \rho_2 := \suppresol_{\rho_1}\sem\E
  $
  we obviously have $\rho^* = \rho_2$.
  We will see that this method can always be applied.
  More precisely, we can always compute $\rho^*$ after performing at most $n$ such iterations.
  \qed
\end{example}

\per
\noindent
In order to deal not only with systems of \cmorcave equations,
but also with systems of \morcave equations,
we use Lemma \ref{l:eindeutige:loesung} iteratively until we reach a solution.
That is, we generalize the statement of Lemma \ref{l:eindeutige:loesung} as follows:

\begin{lemma}
  \label{l:eindeutige:loesung:iter}
  \per
  Let $\E$ be a system of \morcave equations 
  and $\rho$ a feasible pre-solution of $\E$.
  For all $i \in \N$,
  let $\suppresol^i_\rho\sem\E$ be defined by
  \begin{align}
    \suppresol^0_\rho\sem\E &:= \rho \\
    \suppresol^{i+1}_\rho\sem\E &:= \suppresol_{\suppresol^i_\rho\sem\E}\sem\E &
    & \text{for all } i \in \N
    .
  \end{align}
  
  \noindent
  Then, the following statements hold:
  \begin{enumerate}
    \item
      $(\suppresol^i_\rho\sem\E)_{i\in\N}$ is an increasing sequence of feasible pre-solutions of $\E$.
    \item
      $\suppresol^i_\rho\sem\E \leq \mu_{\geq\rho}\sem\E$ for all $i \in \N$.
    \item
      $\suppresol^{\abs X}_\rho\sem\E = \mu_{\geq\rho}\sem\E$.
    \item
      $\suppresol_\rho\sem\E = \mu_{\geq\rho}\sem\E$, whenever $\E$ is a system
      of $\cmorcave$ equations.
  \end{enumerate}
\end{lemma}

\begin{proof}
  \per
  The first two statements can be proven by induction on $i$ using 
  Lemma \ref{l:eindeutige:loesung}.
  The third statement follows from the fact that,
  for any feasible pre-solution $\rho$ of a system $\E$ of \morcave equations,
  $\suppresol_\rho\sem\E < \mu_{\geq\rho}\sem\E$
  implies that there exists some variable $\vx \in \vX$
  such that $\rho(\vx) < \infty$ and $\suppresol_\rho\sem\E(\vx) = \infty$.
  The fourth statement is the second statement of Lemma \ref{l:eindeutige:loesung}.
  \qed
\end{proof}

\begin{example}
  \per
  For the situation in Example \ref{ex:strict:ineq:rho:star},
  we have 
  $ 
    \mu_{\geq\rho} \sem{\E} = \suppresol^3_\rho\sem\E = \suppresol^2_\rho\sem\E > \suppresol^1_\rho\sem\E > \rho
    .
  $ 
  \qed
\end{example}

\per
\noindent
Because of the definition of $\suppresol_\rho$ (see Definition \ref{d:supsol}),
Lemma \ref{l:eindeutige:loesung:iter} implies the following corollary:

\begin{corollary}
  \label{c:depends:on:e:and:x:only}
  \per
  Let $\E$ be a system of \morcave equations and $\rho$ a feasible pre-solution of $\E$.
  Then,
  the value $\mu_{\geq\rho}\sem\E$ only depends on $\E$ and 
  $\vX^\infty_\rho := \{ \vx \mid \vx = e \in \E ,\; \sem{e} \rho = \infty \}$.
  \qed
\end{corollary}

\subsection{Termination}

\per
\noindent
It remains to show that our 
$\vee$-strategy improvement algorithm (Algorithm \ref{alg:alg:stratimp}) terminates.
That is,
we have to come up with an upper bound on the number of iterations of the loop.
In each iteration, 
we have to compute $\mu_{\geq\rho}\sem{\E(\sigma)}$,
where $\rho$ is a feasible pre-solution of $\E(\sigma)$.
This has to be done until we have found a solution.
By Corollary \ref{c:depends:on:e:and:x:only}, 
$\mu_{\geq\rho}\sem{\E(\sigma)}$
only depends on the $\vee$-strategy $\sigma$ 
and the set $\vX^\infty_\rho := \{ \vx \mid \vx = e \in \E(\sigma) ,\; \sem{e} \rho = \infty \}$.
During the run of our $\vee$-strategy improvement algorithm,
the set $\vX^\infty_\rho$ 
monotonically increases.
This implies that 
we have to consider each $\vee$-strategy $\sigma$ at most $\abs \vX$ times.
That is, 
the number of iterations of the loop is bounded from above by $\abs\vX \cdot \abs\MaxStrat$.
Summarizing, we have shown our main theorem:

\begin{theorem}
  \label{t:main}
  \per
  Let $\E$ be a system of \maxmorcave equations in standard form.
  Our $\vee$-strategy improvement algorithm computes $\mu\sem\E$ and
  performs at most $|\vX| \cdot |\MaxStrat|$ $\vee$-strategy improvement steps.
  \qed
\end{theorem}

\per
\noindent
In our experiments, 
we did not observe the exponential worst-case behavior.
All examples we know of require linearly many $\vee$-strategy improvement steps.
We are also not aware of a class of examples,
where we would be able to observe the exponential worst-case behavior.
Therefore, our conjecture is that for practical examples our algorithm terminates after linearly many iterations.

%% file: para_opt.tex
\np
\section{Parametrized Optimization Problems as Right-hand sides}
\label{s:param:opt:rhs}

\per
\noindent
In the static program analysis application that we discuss in Section \ref{s:relaxed:sem}, 
the right-hand sides of the fixpoint equation systems we have to solve
are maxima of finitely many \emph{parametrized optimization problems}.
In this special situation, 
we can evaluate $\vee$-strategies more efficiently 
than by solving general convex optimization problems as described in Section \ref{s:strat:imp}
(see Lemma \ref{l:convex:opt}, \ref{l:eindeutige:loesung}, and \ref{l:eindeutige:loesung:iter}).
We provide a in-depth study of this special situation in this section.

\subsection{
  Parametrized Optimization Problems}

\per
\noindent
We now consider the case that
a system $\E$ of fixpoint equations is given,
where the right-hand sides are \emph{parametrized optimization problems}.
In this article,
we call an operator $g : \CR^n \to \CR$ a \emph{parametrized optimization problem}
if and only if
\begin{align}
  \label{param:right:hand:side}
  g(x) &= \sup\,\{ f(y) \mid y \in Y(\vx_1,\ldots,\vx_n) \} && \text{for all }  x \in \CR^n
  ,
\end{align}

\per
\noindent
where
    $f : \R^{k}\to\R$ is an objective function, and
    $Y : \CR^n \to 2^{\R^{k}}$ 
    is a mapping that assigns a set $Y(x) \subseteq \R^k$ of states to any vector of bounds $x \in \CR^n$.
    The parametrized optimization problem $g$ is monotone on $\CR^n$,
    whenever $Y$ is monotone on $\CR^n$.
    It is monotone on $\CR^n$ 
    and upward chain continuous on $g^{-1}(\CR\setminus\{\neginfty\})$%
    \footnote{A monotone function $g :  \CR^n\to\CR$ is 
    upward chain continuous on an upward closed set $X \subseteq \CR^n$ if and only if $g(\bigvee C) = \bigvee g(C)$ 
    for all non-empty chains $C \subseteq X$.}
    whenever $f$ is continuous on $\R^k$ and 
    $Y$ is monotone  on $\CR^n$ and upward chain continuous on $Y^{-1}(2^{\R^k} \setminus \{ \emptyset \} )$%
    \footnote{A monotone function $Y : \CR^n \to 2^{\R^k}$ is upward chain continuous on an upward closed set $X \subseteq \CR^n$ 
    if and only if 
    $Y(\bigvee C) = \bigcup Y(C)$ for all chains $C \subseteq X$.}%
    .
    In the following,
    we are concerned with the latter situation.
    A parametrized optimization problem $g$ that is monotone on $\CR^n$ 
    and upward chain continuous on $g^{-1}(\CR\setminus\{\neginfty\})$
    is called \emph{upward chain continuous parametrized optimization problem}.


\begin{example}
\per%
	Assume that $Y$ and $f$ are given by
	\begin{align}
	  Y(x) &:= \{ y \in \R^k \mid Ay \leq x \} && \text{for all } x \in \CR^n, \text{ and} \\
	  f(y) &:= b + c^\top y && \text{for all } y \in \R^k
	  ,
	\end{align}
	
     \per
	\noindent
	where $A \in \R^{n\times k}$, $b \in \R$, and $c \in \R^k$.
	Then, $g$ is defined through Equation \eqref{param:right:hand:side} 
	is an upward chain continuous parametrized optimization problems.
	To be more precise,
	it is a
	\emph{parametrized linear programming problem} (to be defined).
	Although this is also an interesting case (cf.\ \citet{DBLP:journals/toplas/GawlitzaS11}),
	in the following, we mainly focus on the more general case where the right-hand sides 
	are \emph{parametrized semi-definite programming problems} (to be defined).
	In this example,
	the right-hand side is not only upward chain continuous,
	it is even \cmcave.
	To be more precise,
	on the set of points where it returns a value greater than $\neginfty$
	it is a point-wise minimum of finitely many monotone and affine operators.
  \qed
\end{example}

\subsection{Fixpoint Equations with Parametrized Optimization Problems}
\label{ss:deal:with:param:opt}

Assume now that we have a system of fixpoint equations,
where the right-hand sides are point-wise maxima of finitely many upward chain continuous parametrized optimization problems.
If we use our $\vee$-strategy improvement algorithm to compute the least solution,
then, 
for each $\vee$-strategy improvement step,
we have to compute $\mu_{\geq\rho_0} \sem{\E}$
for a system $\E$ of fixpoint equations whose right-hand sides are 
upward chain continuous parametrized optimization problems,
and $\rho_0$ is a pre-solution of $\E$.
We study this case in the following:

\per
Assume that $\E$ is a system of fixpoint equations,
where the right-hand sides are \emph{upward chain continuous parametrized optimization problems}.
For simplicity and without loss of generality,
we additionally assume that a variable assignment $\rho_0 : \vX\to\R$ is given such that 
\begin{align}
  \neginftyvar \ll \rho_0 \leq \sem\E\rho_0 \ll \inftyvar
  .
\end{align}

\per
\noindent
We are interested in computing the pre-solution $\suppresol_{\rho_0}\sem\E$ of $\E$.
In the case at hand,
this means that we need to compute $\rho^* : \vX\to\CR$ that is defined by
\begin{align}
  \rho^*(\vx) &:= \sup \, \{ \rho(\vx) \mid \rho : \vX\to\R \text{ and } \rho \leq \sem\E\rho \}
  && \text{ for all } \vx \in \vX
  .
\end{align}



\per
\subsubsection{Algorithm \EvalForMaxAtt}

\noindent
As a start, 
we firstly consider the case where all right-hand sides are 
upward chain continuous parametrized optimisation problems
of the form
$\sup\,\{ f(y) \mid y \in Y(\vx_1,\ldots,\vx_n) \}$,
where
\begin{align} \sup\,\{ f(y) \mid y \in Y(x_1,\ldots,x_n) \} = \max\,\{ f(y) \mid y \in Y(x_1,\ldots,x_n) \} \end{align}

\per
\noindent
for all $x_1,\ldots,x_n \in \CR$ with $\neginfty < \sup\,\{ f(y) \mid y \in Y(x_1,\ldots,x_n) \} < \infty$.
We say that such a parametrized optimization problem \emph{attains its optimal value for all parameter values}.
Parametrized linear programming problems,
for instance,
are parametrized optimization problems
that attain their optimal values for all parameter values.
In the case at hand,
the variable assignment $\rho^*$ can be characterized as follows:
\begin{align}
  \label{eq:chara:system:ineq}
  \rho^*(\vx) &:= \sup\,\{ \rho(\vx) \mid \rho : \vX_{\C(\E)} \to\R \text{ and } \rho \leq \sem{\C(\E)}\rho \}
  &&
  \text{for all } \vx \in \vX
  ,
\end{align}

\per
\noindent
where the constraint system $\C(\E)$ is obtained from $\E$ by replacing every equation 
\begin{align}
  \vx = \sup\;\{ f(y) \mid y \in Y(\vx_1,\ldots,\vx_n) \}
\end{align}

\per\noindent
with the constraints
\begin{align}
  \vx &\leq f(\vy_1,\ldots,\vy_{k}) & 
  (\vy_1,\ldots,\vy_{k}) &\in Y(\vx_1,\ldots,\vx_n)
  ,
\end{align}

\per\noindent
where $\vy_1,\ldots,\vy_{k}$ are fresh variables.

\per
As we will see in the remainder of this section,
the above characterization enable us to compute $\rho^*$ using specialized convex optimization techniques.
If, for instance, the right-hand sides are \emph{parametrized linear programming problems} (to be defined),
then we can compute $\rho^*$ through linear programming.
Likewise,
if the right-hand sides are \emph{parametrized semi-definite programming problems} (to be defined),
then we can compute $\rho^*$ through semi-definite programming.

\begin{example}
  \label{ex:param:easy:case}
  \per
  Let us consider the system $\E$ of equations that consist of the following equations:
  \begin{align}
    \label{eq:to:be:replaced}
    \vx_1 &= \sup \; \{ x_1' \in \R \mid x_1' \in \R ,\; x_1' \leq 0 \} \\
    \vx_2 &=
      \sup \; 
      \{ x_2'' \in \R \mid x_2',x_2'' \in \R ,\; 0 \leq x_2' \leq \vx_1 ,\; x_2'' \leq 1 \} &
  \end{align}
  
  \per
  \noindent
  We aim at computing the variable assignment $\rho^* : \vX\to\CR$ defined by 
  \begin{align}
    \label{rho:star:das:sich:nicht:aendert}
    \rho^*(\vx) &:= 
      \sup \, \{ \rho(\vx) \mid \rho : \vX\to\R \text{ and } \rho \leq \sem\E\rho \}
    && \text{ for all } \vx \in \vX
    .
  \end{align}
  
  \per
  \noindent  
  All right-hand sides of the equations are 
  upward continuous parametrized optimization problems 
  that attain their optimal value for all parameter values.
  Hence, we can apply the above described method to compute $\rho^*$.
  If we do so, the system $\C(\E)$ of inequalities consist of the following inequalities:
  \begin{align}
    \vx_1 &\leq \vx_1' & 
    \vx_1' &\leq 0 & 
    \vx_2 &\leq \vx_2'' & 
    0 &\leq \vx_2' \leq \vx_1 &
    \vx_2'' \leq 1
    .
  \end{align}
  
  \ok
  \noindent
  According to Equation \eqref{eq:chara:system:ineq},
  for all $i \in \{1,2\}$,
  we thus have 
  \begin{align}
    \rho^*(\vx_i) &= \sup \; \{ 
      \vx_i 
      \mid 
      \vx_1, \vx_1',\vx_2,\vx_2',\vx_2'' \in \R, 
      \\&\qquad
      \vx_1 \leq \vx_1' ,
      \vx_1' \leq 0 , 
      \vx_2 \leq \vx_2'' ,
      0 \leq \vx_2' \leq \vx_1 ,
      \vx_2'' \leq 1
    \}
  \end{align}

  \per
  \noindent
  Observe that these optimization problems are actually linear programming problems.
  Solving these linear programming problems gives us, as desired, 
  $\rho^* = \{ \vx_1\mapsto 0 ,\, \vx_2\mapsto 1 \}$.
  \qed
\end{example}

\per
\subsubsection{Algorithm \EvalForGen}

\noindent
If we are not in the nice situation that all 
parametrized optimization problems attain their optimal values for all parameter values,
then we have to apply a more sophisticated method to compute $\rho^*$.
The following example, that is obtained from Example \ref{ex:param:easy:case},
illustrates the need for more sophisticated methods.

\begin{example}
  \label{ex:param:easy:case:not:general}
  \per
  We now slightly modify the fixpoint equation system $\E$ from Example \ref{ex:param:easy:case}
  by replacing Equation \eqref{eq:to:be:replaced} by the equation
  $
    \vx_1 = \sup \; \{ x_1' \in \R \mid x_1' \in \R ,\; x_1' < 0 \}     
  $.
  That is, 
  we are now concerned with strict inequality instead of non-strict inequality.
  In consequence, 
  the parametrized optimization problem does not attain its optimal value for any parameter value.
  The fixpoint equation system $\E$ now consists of the following equations:
  \begin{align}
    \vx_1 &= \sup \; \{ x_1' \in \R \mid x_1' \in \R ,\; x_1' < 0 \} \\
    \vx_2 &=
      \sup \; 
      \{ x_2'' \in \R \mid x_2',x_2'' \in \R ,\; 0 \leq x_2' \leq \vx_1 ,\; x_2'' \leq 1 \} &
  \end{align}
  
  \noindent
  This modification does not change the value of $\rho^*$ 
  (defined by Equation \eqref{rho:star:das:sich:nicht:aendert}),
  since the right-hand side of the first equation still evaluates to $0$.
  However, the system $\C(\E)$ of inequalities is now given by
  \begin{align}
    \vx_1 &< \vx_1' & 
    \vx_1' &\leq 0 & 
    \vx_2 &\leq \vx_2'' & 
    0 &\leq \vx_2' \leq \vx_1 &
    \vx_2'' \leq 1
    .  
  \end{align}
  
  \per
  \noindent
  Since the above inequalities imply $0 \leq \vx_2' \leq \vx_1 < \vx_1' \leq 0$ and thus $0 < 0$,
  there is no solution to the above inequalities.
  Therefore, 
  we cannot apply the methods we applied in Example \ref{ex:param:easy:case}
  to compute $\rho^*$.
  \qed
\end{example}

\per\noindent
We now describe a more sophisticated method to compute $\rho^*$.
For all variable assignments $\rho_0$ and $\rho$,
we define the system $\E_{\rho_0,\rho}$ of equations as follows:
\begin{align}
  \nonumber
  \E_{\rho_0,\rho} 
  &:= 
  \left\{ \vx = \rho_0(\vx) \mid \vx = e \in \E \text{ and } \rho_0(\vx) \geq \sem e \rho \right\}
  \\
  &\qquad\quad \cup \left\{ \vx = e \mid \vx = e \in \E \text{ and } \rho_0(\vx) < \sem e \rho \right\}
\end{align}

\ok\noindent
That is,
$\E_{\rho_0,\rho}$
contains all equations $\vx=e$ of $\E$ whose right-hand sides $e$ 
evaluate under $\rho$ to a value greater than $\rho_0(\vx)$.
The other equations of $\E$ are replaced by $\vx=\rho_0(\vx)$.
%
We again assume that $\rho_0$ is a variable assignment with 
$\neginftyvar \ll \rho_0 \leq \sem{\E}\rho_0 \ll \inftyvar$.
For all $k \in \N_{>0}$, we then define the variable assignment $\rho_k$ inductively by
\begin{align}
  \rho_k(\vx) &:= \sup\,\{ \rho(\vx) \mid \rho : \vX \to \R \text{ and } \rho \leq \sem{\E_{\rho_0,\rho_{k-1}}}\rho \} && \text{ for all } \vx\in\vX
  .
\end{align}

\noindent
Now, $\rho^*$ is the limit of the sequence $(\rho_k)_{k\in\N}$ and the sequence 
reaches its limits after at most $\abs\vX$ steps:

\begin{lemma}
  \label{l:seq:appro:param:opt}
  \per
  The sequence $(\rho_{k})_{k\in\N}$ of variables assignments is increasing,
  $\rho_k \leq \rho^*$ for all $k \in \N$, 
  $\rho_{k+1} > \rho_k$ if $\rho_k < \rho^*$,
  and
  $\rho_{\abs \vX} = \rho^*$.
  Moreover, 
  $\rho_k(\vx) = \sup\,\{ \rho(\vx) \mid \rho : \vX_{\C(\E_{\rho_0,\rho_{k-1}})} 
    \to\R \text{ and } \rho \leq \sem{\C(\E_{\rho_0,\rho_{k-1}})}\rho \}$
  for all $k$ and all $\vx \in \vX$.  
  \qed
\end{lemma}

%
%

\begin{example}
  \label{ex:param:general:case}
  \per
  Let us again consider the fixpoint equation system $\E$ from Example \ref{ex:param:easy:case:not:general}.
  We again aim at computing the variable assignment $\rho^* : \vX\to\CR$ that is defined by 
  $
    \rho^*(\vx) := 
      \sup \, \{ \rho(\vx) \mid \rho : \vX\to\R \text{ and } \rho \leq \sem\E\rho \}
   $
   for all $\vx \in \vX$.
  %
%
  Since 
  $\neginftyvar \ll \rho_0 \leq \sem\E\rho_0 \ll \inftyvar$
  for 
  $\rho_0 := \{ \vx_1\mapsto 0 ,\; \vx_2 \mapsto 0 \}$,
  we can apply the method we just developed.
  The system $\E_{\rho_0,\rho_0}$ is given by
  \begin{align}
    \vx_1 &= 0 & 
    \vx_2 &= 
      \sup \; 
      \{ x_2'' \in \R \mid x_2',x_2'' \in \R ,\; 0 \leq x_2' \leq \vx_1 ,\; x_2'' \leq 1 \} 
  \end{align}
  
  \per
  \noindent
  Therefore, the constraint system $\C(\E_{\rho_0,\rho_0})$ is given by
  \begin{align}
    \vx_1 &\leq 0 & 
    \vx_2 &\leq \vx_2'' & 
    0 &\leq \vx_2' \leq \vx_1 &
    \vx_2'' \leq 1
  \end{align}
  
  \per
  \noindent
  Solving the optimization problems that aims at maximizing $\vx_1$ and $\vx_2$, respectively,
  we get
  $\rho_1 = \{ \vx_1\mapsto 0 ,\; \vx_2\mapsto 1 \}$.
  We then construct the fixpoint equation system $\E_{\rho_0,\rho_1}$.
  The system $\E_{\rho_0,\rho_1}$ is equal to the system $\E_{\rho_0,\rho_0}$, 
  and thus $\C(\E_{\rho_0,\rho_1})$ is equal to the system $\C(\E_{\rho_0,\rho_0})$.
  Therefore,
  we get $\rho^* = \rho_1$ by Lemma \ref{l:seq:appro:param:opt}.
  \qed
\end{example}

%

\per
\subsubsection{Algoritm \EvalForCmorcave}

\noindent
In our static program analysis application we discuss in the next section,
we have the comfortable situation that our right-hand sides are not 
only upward continuous parametrized optimization problems, 
but they are additionally \cmcave.
We can utilize this in order to simplify the above developed procedure \EvalForGen.
The following lemma is the key ingredient for this optimization:

\begin{lemma}
  \label{l:identify:increasing}
  \per
  Let $\rho$ be a feasible pre-solution of a system $\E$ of \cmorcave equations.
  For all $\vx \in \vX$,
  we have
  $\mu_{\geq\rho}\sem\E(\vx) > \rho(\vx)$
  if and only if
  $({\sem\E}^{\abs{\vX}} \rho) (\vx) > \rho(\vx)$.
\end{lemma}

\begin{proof}[Sketch]
  \per
  Since $\rho$ is a feasible pre-solution of $\E$,
  we can w.l.o.g.\ assume that $\sem{e}\rho > \neginfty$ for all equations $\vx = e$ of $\E$.
  Therefore, 
  $\sem\E$ is upward chain continuous on $(\vX\to\CR)_{\geq \rho}$.
  The statement finally follows from the fact that $\sem\E$ is additionally monotone and order-concave.
  \qed
\end{proof}

\noindent
Assume now that we want to use our $\vee$-strategy improvement algorithm to compute 
the least solution of a system of \maxcmorcave equations.
In each $\vee$-strategy improvement step, 
we are then in the situation that we have to compute $\rho^* := \mu_{\geq\rho}\sem\E$,
where $\rho$ is a feasible pre-solution of a system $\E$ of \morcave equations
(cf.\ Lemma \ref{l:alg:sequence:feasible}).
By Lemma \ref{l:identify:increasing},
we can compute the set
\begin{align}
  \label{eq:Xstrich:morcave}
  \vX' := \{ \vx \in \vX \mid \rho^*(\vx) > \rho(\vx)  \}
\end{align}

\noindent
by performing $\abs\vX$ Kleene iteration steps.
We then construct the equation system 
\begin{align}
  \E' := \{ \vx = e \in \E \mid \vx \in \vX' \} \cup \{ \vx = \rho(\vx) \mid \vx \in \vX\setminus\vX'\}
\end{align}

\noindent
By construction,
we get:
\begin{lemma}
  \label{l:eval:for:cmorcave}
  $\rho^*(\vx) = \sup\,\{ \rho(\vx) \mid \rho : \vX_{\C(\E')} \to \R \text{ and } \rho \leq \sem{\C(\E')}\rho \}$ for all $\vx\in\vX$.
  \qed
\end{lemma}

\per\noindent
In consequence,
we can compute $\rho^*$ by performing $\abs\vX$ Kleene iteration steps followed by 
solving $\abs\vX$ optimization problems.

\begin{example}
  \ok
  Let us again consider the fixpoint equation system $\E$ from 
  Example~\ref{ex:param:easy:case:not:general} 
  and \ref{ex:param:general:case}.
  That is,
  $\E$ consists of the following equations:
  \begin{align}
    \vx_1 &= \sup \; \{ x_1' \in \R \mid x_1' \in \R ,\; x_1' < 0 \} \\
    \vx_2 &=
      \sup \; 
      \{ x_2'' \in \R \mid x_2',x_2'' \in \R ,\; 0 \leq x_2' \leq \vx_1 ,\; x_2'' \leq 1 \} &
  \end{align}

  \noindent
  The fixpoint equation system $\E$ is a system of \cmorcave equations.
  The pre-solution $\rho := \{ \vx_1\mapsto 0, \vx_2 \mapsto 0\}$ of $\E$ is feasible.
  Moreover, we have $\neginftyvar \ll \rho \leq \sem\E\rho \ll \inftyvar$.
  We aim at computing $\rho^* := \mu_{\geq \rho} \sem\E$.
  
  We have $\sem{\E}^{\abs\vX} \rho = \sem{\E}^{2} \rho = \{ \vx_1\mapsto 0, \vx_2 \mapsto 1\}$.
  By Lemma \ref{l:identify:increasing},
  we thus get 
  $\vX' := \{ \vx \in \vX \mid \rho^*(\vx) > \rho(\vx)  \} = \{ \vx_2 \}$
  (cf.\ \eqref{eq:Xstrich:morcave}).
  Lemma \ref{l:eval:for:cmorcave} finally gives us 
  \begin{align}
    \rho^*(\vx_i) 
    =
    \sup\;\{
      \vx_i \mid
      \vx_1,\vx_1',\vx_2,\vx_2',\vx_2'' \in \R,
      \vx_1 \leq 0 ,
      \vx_2 \leq \vx_2'' ,
      0 \leq \vx_2' \leq \vx_1 ,
      \vx_2'' \leq 1
    \}
    .
  \end{align}
  
  \noindent
  for all $i \in \{1,2\}$ (cf.\ Example \ref{ex:param:general:case}).
  This is the desired result.
  We performed two Kleene iteration steps and 
  solved two mathematical optimization problems.
  \qed
\end{example}

\subsection{Parameterized Linear Programming Problems}

\per
\noindent
We now introduce \emph{parameterized linear programming problems}.
We do this as follows.
  For all
  $A \in \R^{k\times m}$ and 
  all $c \in \R^m$,
  we define the operator
  $\LPAc : \CR^k \to \CR$ 
  which solves a \emph{parametrized} linear programming problem
  by
  \begin{align}
    \LPAc(b)
    &:=
  \sup \,
  \{ 
  c^\top x 
  \mid
  x \in \R^m
  \text{ and }
  A x \leq b
  \}
  && \text{for all }
  b \in \CR^k
  .
\end{align}

\noindent
We use the LP-operators in the right-hand sides of fixpoint equation systems:

\begin{definition}(LP-equations, $\vee$-LP-equations)
  A fixpoint equation $\vx=e$ is called \emph{LP-equation}
  if and only if
  $e$ is a parametrized linear programming problem.
  It is called \emph{$\vee$-LP-equation}
  if and only if 
  $e$ is a point-wise maximum of finitely many semi-definite programming problems.
  \qed
\end{definition}

\per
\noindent
LP-operators have the following important properties:

\begin{lemma}
  \label{l:lp:order:conc}
  \per
  The following statements hold for all $A \in \R^{k\times m}$ and all $c \in \R^m$:
    \begin{enumerate}
      \item
        The operator $\LPAc$ is \cmcave.
      \item
        $\LPAc(b) = \max \, \{   c^\top x \mid x \in \R^m \text{ and } A x \leq b\}$
        for all
        $b \in \CR^k$
        with
        $\neginfty < \LPAc(b) < \infty$.
        That is,
        the parametrized optimization problem $\LPAc$ attains its optimal value for all parameter values.
    \end{enumerate}
\end{lemma}

\begin{proof}
   \per
      We do not prove the first statement,
      since, as we will see, it is just a special case of Lemma \ref{l:sdp:order:conc} (see below).
      This second statement is a direct consequence of the fact that the optimal value of a feasible and 
      bounded linear programming problem 
      is attained at the edges of the feasible space.
  \qed
\end{proof}

\ok
\noindent
If we apply our $\vee$-strategy improvement algorithm for solving a system of $\vee$-LP-equations,
then, because of Lemma \ref{l:lp:order:conc},
we have the convenient situation that we can apply Algorithm \EvalForMaxAtt 
instead of its more general variant \EvalForGen
for evaluating a single $\vee$-strategy that is encountered during the $\vee$-strategy iteration
(see Section \ref{ss:deal:with:param:opt}).
We thus obtain the following result:

\begin{theorem}
  \per
%
   If $\E$ is a system of $\vee$-LP-equations,
   then the evaluation of a $\vee$-strategy that is encountered during the $\vee$-strategy iteration can be performed 
   by solving $\abs{\vX}$ linear programming problems,
   each of which can be constructed in polynomial time.   
   In consequence,
   a $\vee$-strategy improvement step can be performed in polynomial time.
   \qed
\end{theorem}

\noindent
Theorem \ref{t:main} 
implies that our $\vee$-strategy improvement algorithm terminates after at most 
$\abs\vX \cdot \abs\Sigma$
$\vee$-strategy improvement steps,
whenever it runs on a system $\E$ of $\vee$-LP-equations.

A consequence of the fact that we can evaluate $\vee$-strategies in polynomial time 
is the following decision problem is in $\mathsf{NP}$:
Decide whether or not,
for a given system $\E$ of $\vee$-LP-equations, 
a given variable $\vx \in \vX$, and
a given value $b \in \CR$,
the statement $\mu\sem\E(\vx) \leq b$ holds.
This decision problem is at least as hard as the problem of computing the winning regions in mean payoff games.
However, whether or not it is $\mathsf{NP}$-hard is an open question.

\subsection{Parameterized Semi-Definite Programming Problems}

\per
\noindent
As a strict generalization of \emph{parameterized linear programming problems},
we now introduce \emph{parameterized semi-definite programming problems}.
Before we can do so, 
we have to briefly introduce semi-definite programming.

\per
\subsubsection{Semi-definite Programming}

$\SR{n}$ 
(resp.\  $\SRp{n}$) 
denotes the set of symmetric matrices
(resp.\ the set of positive semidefinite matrices).
$\preceq$ denotes the Löwner ordering of symmetric matrices,
i.e., $A \preceq B$ if and only if $B - A \in \SRp{n}$.
        $\Tr(A)$ denotes the trace of a square matrix $A \in \R^{n \times n}$, 
        i.e., $\Tr(A) = \sum_{i = 1}^n A_{i \cdot i}$.
        The inner product of two matrices $A$ and $B$ is denoted by 
        $A \bullet B$,
        i.e., $A \bullet B = \Tr (A^\top B)$.
        For $\A = (A_1,\ldots,A_m)$
        with $A_i \in \R^{n \times n}$ for all $i =1,\ldots,m$,
        we denote the vector 
        $(A_1 \bullet X , \ldots , A_m \bullet X)^\top$
        by
        $\A(X)$.
        For all $x \in \R^n$,
        the dyadic matrix $X(x)$
        is defined by
       \begin{align}
          X(x) 
          := 
          \begin{pmatrix}1 \\ x \end{pmatrix} (1, x^\top)
          .
       \end{align}
       

\ok
\noindent
We consider semidefinite programming problems (SDP problems for short)
of the form 
\begin{align}
  z^*
  =
  \sup 
  \;
  \{ 
  C \bullet  X \mid 
  X \in \SRp n
  ,
  \A(X) = a
  ,
  \B(X) \leq b
  \}
  ,
\end{align}

\noindent
  where 
  $\A = (A_1,\ldots,A_m)$,
  $a \in \R^m$,
  $A_1,\ldots,A_m \in \SR n$,
  $\B = (B_1,\ldots,B_k)$,
  $B_1,\ldots,B_k \in \SR n$,
  $b \in \R^k$, and
  $C \in \SR n$.
  The set 
 $
  \{ 
  X \in \SRp n
  \mid
  \A(X) = a
  ,
  \B(X) \leq b
  \}
  $
  is called the \emph{feasible space}.
  The problem is called \emph{feasible} if and only if the feasible space is non-empty.
  It is called \emph{infeasible} otherwise.
  An element of the feasible space
  is called \emph{feasible solution}.
  The value $z^*$ is called \emph{optimal value}.
  The problem is called \emph{bounded} iff $z^* < \infty$. 
  It is called \emph{unbounded}, otherwise.
  A feasible solution $X^*$ is called an \emph{optimal solution} if and only if $z^* = C \bullet X^*$.
  In contrast to the situation for linear programming,
  there exist feasible and bounded semi-definite programming problem 
  that have no optimal solution.
  
  For semi-definite programming problems,
  fast algorithms exist. 
  Semi-definite programming is polynomial time solvable 
  if an a priori bound on the size of the solutions is known and provided as an input.
  
  For more detailed information on semi-definite programming, or, more generally, on convex optimization,
  we refer, for instance, to
  \citet{todd01,nemirovski05}.

\per
\subsubsection{Parametrized SDP Problems}

\noindent
  For 
  $\A = (A_1,\ldots,A_m)$,
  $A_1,\ldots,A_m \in \SR n$,
  $a \in \R^m$, 
  $\B = (B_1,\ldots,B_k)$,
  $B_1,\ldots,\allowbreak B_k \in \SR n$, and
  $C \in \SR n$,
  we define the operator
  $\SDPAaBC : \CR^k \to \CR$ 
  which solves a \emph{parametrized} SDP problem
  by
  \begin{align*}
    \SDPAaBC(b)
    &:=
  \sup \,
  \{ 
  C {\bullet}  X \mid 
  X {\in} \SRp n
  ,
  \A(X) = a
  ,
  \B(X) \leq b
  \}
  && 
  \!\!\!\!
  \text{for all }
  b \in \CR^k
  .
  \end{align*}
  
\per
\noindent
The SDP-operators generalizes the LP-operators
in the same way as semi-definite programming generalizes linear programming.
That is,
for every LP-operator we can construct an equivalent SDP-operator.

\begin{definition}(SDP-equations, $\vee$-SDP-equations)
  A fixpoint equation $\vx=e$ is called \emph{SDP-equation}
  if and only if
  $e$ is a parametrized semi-definite programming problem.
  It is called \emph{$\vee$-SDP-equation}
  if and only if 
  $e$ is a point-wise maximum of finitely many semi-definite programming problems.
  \qed
\end{definition}

\noindent
For this article,
the following properties of SDP-operators are important:

\begin{lemma}
  \per
  \label{l:sdp:order:conc}
  The operator 
  $\SDPAaBC$ 
  is \cmcave.
\end{lemma}

\begin{proof} 
  \ok
  Let $f := \SDPAaBC$.
  For all $b \in \R^k$,
  let $M(b) := \{ X \in \SRp n \mid \A(X) = a, \B(X) \leq b \}$.
  Therefore,
  $f(b) = \sup\;\{ C \bullet X \mid X \in M(b) \}$ for all $b \in \R^k$.
  We do not need to consider all $I : \{1,\ldots,k\} \to \{\neginfty,\mathsf{id},\infty\}$,
  because,
  for all $I : \{1,\ldots,k\} \to \{\neginfty,\mathsf{id},\infty\}$,
  $f^{(I)}$ can be obtained by choosing appropriate $\mathcal A, a, \mathcal B, C$.  
  The fact that $f$ is monotone is obvious.
  Firstly, we show that $f(b) < \infty$ holds for all $b \in \R^k$,
  whenever $\fdom(f) \neq \emptyset$.
  For the sake of contradiction assume 
  that there exist $b_1,b_2 \in \R^k$
  such that
  $f(b_1) \in \R$ and $f(b_2) = \infty$ hold.
  Note that 
  $M(b_i)$
  are convex sets for all $i \in \{1,2\}$.
  Thus, there exists some $D \in \SRp n$ such that
  $C \bullet D > 0$ and 
  $M(b_2) + \{ \lambda D \mid \lambda \in \Rp \} \subseteq M(b_2)$
  hold.
  Therefore, $\A(D) = 0$ and $\B(D) \leq 0$.
  Let $X_1 \in \SRp n$ with $\A(X_1) = a$ and $\B(X_1) \leq b_1$.
  Then 
  $
    \A(X_1 + \lambda D) 
    =
    \A(X_1) + \lambda \A(D)
    =
    a
  $
  and
  $
    \B(X_1 + \lambda D)
    =
    \B(X_1) + \lambda \B(D)
    \leq b_1
  $
  hold for all $\lambda > 0$.
  Thus, $f(b_1) = \infty$ 
  --- contradiction.
  Thus, $f(b) < \infty$ holds for all $b \in \R^k$,
  whenever $\fdom(f) \neq \emptyset$.
  
  \ok
  Next,
  we show that $\fdom(f)$ is convex and $f|_{\fdom(f)}$ is concave.
  Assume that $\fdom(f) \neq \emptyset$.
  Thus, $f(b) < \infty$ for all $b \in \R^k$.
  Let 
  $b_1, b_2 \in \fdom(f)$,
  $\lambda \in [0,1]$, and $b := \lambda b_1 + (1-\lambda) b_2$.
  In order to show that
  \begin{align}
  \label{eq:set:inc}
    \lambda M(b_1) + (1 - \lambda) M(b_2)
    \subseteq
    M(b)
  \end{align}
  holds,
  let 
  $X_i \in M(b_i)$, $i =1,2$, and
  $X = \lambda X_1 + (1 - \lambda) X_2$.
  Since
  $X_i \in \SRp n$,
  $\A(X_i) = a$, and 
  $\B(X_i) \leq b_i$
  for all $i=1,2$,
  we have 
  $X \in \SRp n$,
  $\A(X) = \lambda\A(X_1) + (1-\lambda)\A(X_2)= a$,
  $
  \B(X) 
  = 
  \lambda \B(X_1) + (1 - \lambda) \B(X_1)
  \leq
  \lambda b_1 + (1-\lambda) b_2
  =
  b
  $.
  Therefore, $X \in M(b)$.
  Using \eqref{eq:set:inc},
  we finally get:
  \begin{align}
    f( b )
    &=
    \sup \{ 
      C \bullet X 
      \mid 
      X \in M(b) 
    \}
    \\&\geq
    \lambda 
    \sup \{ C \bullet X_1 \mid X_1 \in M(b_1) \}
    +
    (1-\lambda) 
    \sup \{ C \bullet X_2 \mid X_2 \in M(b_2) \}
    \\&=
    \lambda f(b_1)
    +
    (1-\lambda) f(b_2)
    > 
    \neginfty
  \end{align}

  \ok
  \noindent
  Therefore, $\fdom(f)$ is convex and $f|_{\fdom(f)}$ is concave.

  \ok
  It remains to show that $f$ is upward chain continuous on $f^{-1}(\CR\setminus\{\neginfty\})$.
  For that, 
  let $B \subseteq f^{-1}(\CR\setminus\{\neginfty\})$ be a chain.
  We have
  \begin{align}
    \textstyle f(\bigvee B) 
    &=
    \textstyle \sup \; \{ C \bullet X \mid X \in M(\bigvee B) \}
    \\&=
    \textstyle \sup \; \{ C \bullet X \mid X \in \bigcup \{ M(b) \mid b \in B \} \} 
    && \text{($M$ is continuous)}
    \\&=
    \textstyle \sup \; \{ \sup \; \{ C \bullet X \mid X \in M(b) \} \mid b \in B \} 
    \\&=
    \textstyle \sup \; \{ f(b) \mid b \in B \} 
  \end{align}
  
  \ok
  \noindent
  This proves that $f$ is upward chain continuous on $f^{-1}(\CR\setminus\{\neginfty\})$.
    \qed
\end{proof}

\per
\noindent
The next example shows that the square root operator
can be expressed through a SDP-operator:

\begin{example}
  \per
  \label{ex:wurzel}
  The square root operator $\sqrt\cdot : \CR\to\CR$ is defined by
  $\sqrt b := \sup \{ x \in \R \mid x^2 \leq b \}$
  for all $b \in \CR$. 
  Note that $\sqrt b = \neginfty$ for all $b < 0$, and
  $\sqrt \infty = \infty$.
  Let
  \begin{align}
    \textstyle
    \A := \begin{pmatrix}\begin{pmatrix} 1 & 0 \\ 0 & 0 \end{pmatrix}\end{pmatrix},
    \quad
    a := 1,
    \quad
    \B := \begin{pmatrix}\begin{pmatrix} 0 & 0 \\ 0 & 1 \end{pmatrix}\end{pmatrix},
    \quad
    C := \begin{pmatrix} 0 & \frac12 \\ \frac12 & 0 \end{pmatrix}. 
  \end{align}
  
  \per
  \noindent
  For $x,b \in \Rp$,
  the statement
  $x^2 \leq b$ is equivalent to the statement $\exists b' . x^2 \leq b' \leq b$.
  By the Schur complement theorem 
  (c.f.\ Section 3, Example 5 of \citet{todd01}, for instance),  
  this is equivalent to 
  \begin{align}
    \exists b' .
    \begin{pmatrix}
      1 & x \\
      x & b'
    \end{pmatrix}
    \succeq
    0
    \wedge
    b' \leq b
    .
  \end{align}
  
  \per
  \noindent
  This is equivalent to 
  $ 
    \textstyle
    \exists X \in \SRp 2 .
    x = X_{1 \cdot 2} = X_{2 \cdot 1}
    \wedge
    \A(X) = a
    \wedge
    \B(X) \leq b
  $. 
  Thus,
  $ 
   \sqrt b = 
    \SDPAaBC(b)
  $ 
  for all $b \in \CR$.
  \qed
\end{example}

\ok
\noindent
If $\E$ is a system of $\vee$-SDP-equations,
then, because of Lemma \ref{l:sdp:order:conc},
we have the convenient situation that we can apply Algorithm \EvalForCmorcave
instead of its more general variant \EvalForGen
(see Section \ref{ss:deal:with:param:opt})
to evaluate the $\vee$-strategies that are encountered during the $\vee$-strategy iteration.
This case is in particular interesting for the static program analysis application we will describe in Section \ref{s:relaxed:sem}.

\begin{theorem}
\label{t:sdp:eqs:strat:imp}
\ok
   If $\E$ is a system of $\vee$-SDP-equations,
   then the evaluation of a $\vee$-strategy that is encountered during the $\vee$-strategy iteration can be performed 
   by performing $\abs\vX$ Kleene iteration steps and subsequently
   solving $\abs{\vX}$ semi-definite programming problems,
   each of which can be constructed in polynomial time.
   \qed
\end{theorem}

\noindent
Theorem \ref{t:main} 
implies that our $\vee$-strategy improvement algorithm terminates after at most 
$\abs\vX \cdot \abs\Sigma$
$\vee$-strategy improvement steps,
whenever it runs on a system $\E$ of $\vee$-LP-equations.

%% file: quadratic_zones.tex
\np\section{Quadratic Zones and Relaxed Abstract Semantics}
\label{s:relaxed:sem}

\ok
\noindent
In this section,
we apply our $\vee$-strategy improvement algorithm to a static program analysis problem.
For that, we first introduce our programming model
as well as its collecting and its abstract semantics.
We then relax the abstract semantics along the same lines as 
\citets{DBLP:conf/esop/AdjeGG10} using Shor's semidefinite relaxation schema.
Finally, we show how we can use our finding to compute the relaxation of the abstract semantics.

\subsection{Collecting Semantics}

In our programming model,
we consider statements of the following two forms:
\begin{enumerate}
  \item
    $x := Ax + b$, 
    where $A \in \R^{n \times n}$, and $b \in \R^n$ \emph{(affine assignments)}
  \item
    $x^\top Ax + 2 b^\top x \leq c$, 
    where $A \in \SR n$, $b \in \R^n$, and $c \in \R$ \emph{(quadratic guards)}
\end{enumerate}

\noindent
Here, $x \in \R^n$ denotes the vector of program variables.
We denote the set of statements by $\Stmt$.
The \emph{collecting semantics} $\sem{s} : 2^{\R^n} \to 2^{\R^n}$ of 
a statement $s \in \Stmt$ is defined by:
\begin{align}
  \sem{x := Ax + b}X 
    &:= \{ Ax + b \mid x \in X \}
    && \text{for all } X \subseteq \R^n
  \\
  \sem{x^\top A x + 2 b^\top x \leq c}X 
    &:= \{ x \in X \mid x^\top A x + 2 b^\top x \leq c \}
    && \text{for all } X \subseteq \R^n
\end{align}

\noindent
\ok
A program $G$ is a triple $(N,E,\start, I)$,
where $N$ is a finite set of \emph{control-points},
$E \subseteq N \times \Stmt \times N$ is a finite set of control-flow edges,
$\start \in N$ is the start control-point, 
and $I \subseteq \R^n$ is a set of initial values.
The \emph{collecting semantics} $\Values$ of a program $G = (N,E,\start, I)$ is then the least solution of 
the following constraint system:
\begin{align}
  \VALUES[\start]
    &\supseteq I
  &
  \VALUES[v] 
    &\supseteq \sem{s} (\VALUES[u])
  && \text{for all } (u,s,v) \in E
\end{align}

\noindent
Here, the variables $\VALUES[v]$, $v \in N$ 
take values in $2^{\R^n}$.
The components of 
the collecting semantics $\Values$ 
are denoted by $\Values[v]$
for all $v \in N$.

\subsection{Quadratic Zones and Abstract Semantics}

Along the lines of 
\citets{DBLP:conf/esop/AdjeGG10},
we define \emph{quadratic zones}
as follows:
%
  A set $P$ of templates $p : \R^n \to \R$ is a 
  \emph{quadratic zone}
  if and only if
  every template $p \in P$ can be written as
  \begin{align}
    p(x) &= x^\top A_p x + 2 b_p^\top x
    && \text{for all } x \in \R^n
    ,
  \end{align}

  \noindent
  where
  $A_p \in \SR n$ and
  $b_p \in \R^n$
  for all $p \in P$.
%
In the remainder of this article, 
we assume that 
$P = \{ p_1,\ldots,p_m \}$ 
is a finite quadratic zone.
Moreover, we assume w.l.o.g.\ that $p_i \neq 0$ for all $i = 1,\ldots,m$.
The \emph{abstraction} $\alpha : 2^{\R^n} \to P \to \CR$ 
and the \emph{concretization} $\gamma : (P \to \CR) \to 2^{\R^n}$ are 
defined as follows:
\begin{align}
  \gamma(v) 
    &\textstyle:= 
    \{ x \in \R^n \mid \forall p \in P . p(x) \leq v(p) \} 
    && \text{for all } v : P \to \CR \\
  \alpha(X)
    &\textstyle:= 
    \bigwedge \{ v : P \to \CR \mid \gamma(v) \supseteq X \}
     && \text{for all } X \subseteq \R^n
\end{align}

\noindent
As shown by \citets{DBLP:conf/esop/AdjeGG10},
$\alpha$ and $\gamma$ form a Galois-connection.
The elements from $\gamma(P \to \CR)$ 
and the elements from $\alpha(2^{\R^n})$
are called \emph{closed}.
$\alpha(\gamma(v))$ is called the \emph{closure of $v : P \to \R$}.
Accordingly, $\gamma(\alpha(X))$ is called the 
\emph{closure of $X \subseteq \R^n$}.


As usual,
the \emph{abstract semantics} $\sem{s}^\sharp : (P \to \CR) \to P \to \CR$ 
of a statement $s$
is defined by $\sem{s}^\sharp := \alpha \circ \sem{s} \circ \gamma$.
The \emph{abstract semantics} $\Values^\sharp$ of a program
$G = (N,E,\start,I)$ is then the least solution of the following 
constraint system:
\begin{align}
  \VALUES^\sharp[\start]
    &\geq \alpha(I)
  &
  \VALUES^\sharp[v] 
    &\geq \sem{s}^\sharp (\VALUES^\sharp[u])
  && \text{for all } (u,s,v) \in E
\end{align}

\noindent
Here, the variables $\VALUES^\sharp[v]$, $v \in N$ 
take values in $P \to \CR$.
The components of 
the abstract semantics $\Values^\sharp$ 
are denoted by $\Values^\sharp[v]$
for all $v \in N$.

\subsection{Relaxed Abstract Semantics}

  \ok
  \noindent
  The problem of deciding,
  whether or not, 
  for a given quadratic zone $P$,
  a given $v : P \to \CQ$,
  a given $p \in P$,
  and a given $q \in \CQ$,
  $\alpha(\gamma(v))(p) \leq q$
  holds, 
  is NP-hard (cf.\ \citet{DBLP:conf/esop/AdjeGG10}) and thus intractable.
  Therefore,
  we use the  
  \emph{relaxed abstract semantics} $\Values^\Relaxed$
  introduced by \citets{DBLP:conf/esop/AdjeGG10}.
  It is based on Shor's semidefinite relaxation schema.
  In order to fit it into our framework,
  we have to switch to the semi-definite dual.
  This is not a disadvantage.
  It is actually an advantage,
  since we gain additional precision through this step.


\begin{definition}[$\sem{x := Ax + b}^\Relaxed$]
  \label{d:relaxation:assignment}
  We define the \emph{relaxed abstract semantics}
  $\sem{x := Ax + b}^\Relaxed : (P \to \CR) \to P \to \CR$
  of an affine assignment $x := Ax + b$
  by
  \begin{align}
    &\sem{x := Ax + b}^\Relaxed v \, (p)
    \\&\qquad:=\textstyle
    \sup \{ \overline A (p) {\bullet} X 
      \mid \forall p' \in P . \overline A_{p'} {\bullet} X \leq v(p'), 
      X \succeq 0, X_{1\cdot 1} = 1  \}
  \end{align}
  
  \noindent
  for all $v : P \to \R$ and all $p \in P$,
  where, for all $p' \in P$, 
  \begin{align}
    A(p) := A^\top A_p A ,\quad
    b(p) := A^\top A_p b + A^\top b_p,\quad
    c(p) := b^\top A_p b + 2 b_p^\top b
 \\
    \overline A (p) 
    :=
    \begin{pmatrix}
      c(p) & b^\top (p) \\
      b(p) & A(p)
    \end{pmatrix}
    ,\quad
    \overline A_{p'}
    := 
    \begin{pmatrix}
      0 & b_{p'}^\top \\
      b_{p'} & A_{p'}
    \end{pmatrix}
    .
    \text\qed
  \end{align}
\end{definition}

\begin{definition}
  [$\sem{x^\top A x + 2 b^\top x \leq c}^\Relaxed$]
  \label{d:relaxation:guards}
  We define the \emph{relaxed abstract semantics}
  $\sem{x^\top A x + 2 b^\top x \leq c}^\Relaxed : (P \to \CR) \to P \to \CR$
  of a quadratic guard $x^\top A x + 2 b^\top x \leq c$ by
  \begin{align}
    &\;\sem{x^\top A x + 2 b^\top x \leq c}^\Relaxed v \, (p)
    \\&\qquad :=\;\textstyle
    \sup \{ \overline A_p {\bullet} X 
      \mid \forall p' \in P . \overline A_{p'} {\bullet} X \leq v(p'), 
      \widetilde A {\bullet} X \leq 0,
      X \succeq 0, 
      X_{1\cdot 1} = 1  \}
  \end{align}
  
  \noindent
  for all $v : P \to \R$ and all $p \in P$,
  where, for all $p' \in P$, 
  \begin{align}
    \widetilde A
    :=
    \begin{pmatrix}
      -c      & b^\top \\
      b & A
    \end{pmatrix}
    ,\quad
    \overline A_{p'} 
    := 
    \begin{pmatrix}
      0 & b_{p'}^\top \\
      b_{p'} & A_{p'}
    \end{pmatrix}
    .
    \text\qed
  \end{align}
\end{definition}

\noindent
The relaxed abstract semantics $\sem{\cdot}^\Relaxed$ is the semidefinite dual of the 
one used by \citets{DBLP:conf/esop/AdjeGG10}.
By weak-duality,
it is at least as precise as the one used by \citets{DBLP:conf/esop/AdjeGG10}.

Next,
we show that 
the relaxed abstract semantics is indeed a relaxation of 
the abstract semantics,
and that the relaxed abstract semantics of a statement is expressible through a SDP-operator.

\begin{lemma}
\per
\label{l:properties:relaxed}
  The following statements hold for every statement $s \in \Stmt$:
  \begin{enumerate}
    \item
      $\sem{s}^\sharp \leq \sem{s}^\Relaxed$
    \item
      For every $i \in \{1,\ldots,m\}$,
      there exist $\A,a,\B,C$ such that
      \begin{align}
        \sem{s}^\Relaxed v \, (p_i) = \SDPAaBC(v(p_1),\ldots,v(p_m))
      \end{align}
      
      \noindent
      for all $v : P \to \CR$.
      From $s$, 
      the values $\A$, $a$, $\B$, and $\C$ can be computed in polynomial time.
      \qed
  \end{enumerate}
\end{lemma}

\begin{proof}
  Since the second statement is obvious, 
  we only prove the first one.
  We only consider the case that 
  $s$ is an affine assignment  $x := Ax + b$.
  The case that $s$ is a quadratic guard 
  can be treated along the same lines.
%
  Let $v : P \to \R$, $p \in P$, and $v' := \sem{x := Ax + b}^\sharp v$.
  Then,
  \begin{align}
    v'(p) 
    &= 
    \sup 
    \{ 
      p(Ax + b) \mid x \in \R^n , \forall p' \in P . p'(x) \leq v(p') 
    \}
    \\&= 
    \sup 
    \{ 
    x^\top A(p) x + 2 b^\top(p) x + c(p)
    \mid
    \\&\qquad\qquad\qquad\qquad
    x \in \R^n
    ,
    \forall p' \in P . x^\top A_{p'} x + 2 b_{p'}^\top x \leq v(p')
    \}
    \\&=\textstyle
    \sup
    \left\{
      (1,x^\top)
      \overline A(p)
      (1,x^\top)^\top
      \mid
      \forall p' \in P
      .
      (1,x^\top)
      \overline A_{p'}
      (1,x^\top)^\top
      \leq 
      v(p')
    \right\}
    \\&=\textstyle
    \sup
    \left\{
      \overline A(p) {\bullet} X(x)
      \mid
      \forall p' \in P
      .
      \overline A_{p'}
      {\bullet}
      X(x)
      \leq 
      v(p')
    \right\}
    \\&\leq\textstyle
    \sup
    \left\{
      \overline A(p) {\bullet} X
      \mid
      \forall p' \in P
      .
      \overline A_{p'}
      {\bullet}
      X
      \leq 
      v(p')
      ,
      X \succeq 0
      ,
      X_{1 \cdot 1} = 1
    \right\}
    .
  \end{align}
  
  \noindent
  The last inequality holds,
  because $X(x) \succeq 0$ 
  and $X(x)_{1 \cdot 1} = 1$
  for all $x \in \R^n$.
  This completes the proof of statement 1.
 \qed
\end{proof}

\noindent
A relaxation of the closure operator 
$\alpha \circ \gamma$
is given by $\sem{x := x}^\Relaxed$.
That is, 
$\alpha \circ \gamma \leq \sem{x := x}^\Relaxed$.


The \emph{relaxed abstract semantics} $\Values^\Relaxed$
of a program $G = (N,E,\start,I)$
is finally defined as the least solution of the following constraint system:
\begin{align*}
  \VALUES^\Relaxed[\start]
    &\geq \alpha(I)
  &
  \VALUES^\Relaxed[v] 
    &\geq \sem{s}^\Relaxed (\VALUES^\Relaxed[u])
  && \text{for all } (u,s,v) \in E
\end{align*}
Here, the variables $\VALUES^\Relaxed[v]$, $v \in N$ 
take values in $P \to \CR$.
The components of 
the relaxed abstract semantics $\Values^\Relaxed$ 
are denoted by $\Values^\Relaxed[v]$
for all $v \in N$.

\per
Because of Lemma \ref{l:properties:relaxed},
the relaxed abstract semantics of a program is a safe over-approximation 
of its abstract semantics.
If all templates and all guards are linear, 
then the relaxed abstract semantics is precise (cf.\ \citet{DBLP:conf/esop/AdjeGG10}):

\begin{lemma}
  We have $\Values^\sharp \leq \Values^\Relaxed$.
  Moreover,
  if all templates and all guards are linear,
  then $\Values^\sharp = \Values^\Relaxed$.
  \qed
\end{lemma}

\subsection{Computing Relaxed Abstract Semantics}
\label{s:sompute:relaxed:semantics}

We now use our $\vee$-strategy improvement algorithm to compute the relaxed abstract semantics 
$\Values^\Relaxed$
of a program $G = (N,E,\start,I)$ w.r.t.\ a given finite quadratic zone $P = \{ p_1,\ldots,p_m \}$.
For that, we define $\C$ to be the constraint system
\begin{align}
   \vx_{\start,p}
    &\geq
    \alpha(I)(p)
    && \text{for all } p \in P
   \\
    \vx_{v,p} 
    &\geq
    (\sem{s}^\Relaxed (\vx_{u,p_1},\ldots,\vx_{u,p_m})^\top )(p)
    && \text{for all } (u, s, v) \in E\text{, and all } p \in P
\end{align}

\noindent
which uses the variables $\vX = \{ \vx_{v,p} \mid v \in N ,\; p \in P \}$.
The value of the variable $\vx_{v,p}$ is the bound 
on the template $p$ at control-point $v$.

Because of Lemma \ref{l:properties:relaxed},
from $\C$ we can construct a system $\E$ of SDP-equations
with $\mu\sem\E = \mu\sem\C$ in polynomial time.
Finally, 
we have:

\begin{lemma}
  $\Values^\Relaxed[v](p) = \mu\sem{\E}(\vx_{v,p})$
  for all $v \in N$ and all $p \in P$.
  \qed
\end{lemma}

\noindent
Since $\E$ is a system of $\vee$-SDP-equations,
by Theorem \ref{t:main} and Theorem \ref{t:sdp:eqs:strat:imp},
we can compute the least solution $\mu\sem\E$ of $\E$
using our $\vee$-strategy improvement algorithm.
Thus, we have finally shown the following main result for the static program analysis application:

\begin{theorem}
  We can compute the relaxed abstract semantics 
   $\Values^\Relaxed$
  of a program $G = (N,E, \allowbreak \start,I)$
  using our $\vee$-strategy improvement algorithm.
  Each $\vee$-strategy improvement step can 
  by performed by 
  performing $\abs N \cdot \abs P$ Kleene iteration steps and 
  solving 
  $\abs N \cdot \abs P$ SDP problems,
  each of which can be constructed in polynomial time.
  The number of strategy improvement steps is exponentially 
  bounded by the product of the number of merge points in the program and the number of program variables.
  \qed
\end{theorem}

\begin{example}
  \ok
  In order to give a complete picture of our method,
  we now discuss the harmonic oscillator example of 
  \citet{DBLP:conf/esop/AdjeGG10} in detail.
  The program consists only of the simple loop
  \begin{align}
    \WHILE (\TRUE)
    \;
    x := Ax
  ,
  \end{align}
  
  \noindent
  where $x = (x_1,x_2)^\top \in \R^2$ 
  is the vector of program variables 
  and
  \begin{align}
    A
    =
    \begin{pmatrix}
      1 & 0.01 \\
      -0.01 & 0.99
    \end{pmatrix}
    .
  \end{align}

  \noindent
  We assume that the two-dimensional interval $I = [0,1] \times [0,1]$ is the set of initial states.
  The set of control-points just consists of $\start$, i.e.\ $N = \{ \start \}$.
  The set $P = \{ p_1,\ldots,p_5 \}$ of templates is given by
  \begin{align}
    p_1(x_1,x_2) &= -x_1 &
    p_2(x_1,x_2) &=  x_1 &
    p_3(x_1,x_2) &= -x_2 \\
    p_4(x_1,x_2) &= x_2 &
    p_5(x_1,x_2) &= 2 x_1^2 + 3 x_2^2 + 2 x_1 x_2
  \end{align}

  \noindent
  The abstract semantics is thus given by the 
  least solution 
  of the following system of $\vee$-SDP-equations:
  \begin{align}
    \vx_{\start,p_1} 
    &= 
    \neginfty \vee 0 \vee \SDP_{\A,a,\B,C_1} (\vx_{\start,p_1},\vx_{\start,p_2},\vx_{\start,p_3},\vx_{\start,p_4},\vx_{\start,p_5})
    \\
    \vx_{\start,p_2} 
    &= 
    \neginfty \vee 1 \vee \SDP_{\A,a,\B,C_2} (\vx_{\start,p_1},\vx_{\start,p_2},\vx_{\start,p_3},\vx_{\start,p_4},\vx_{\start,p_5})
    \\
    \vx_{\start,p_3} 
    &= 
    \neginfty \vee 0 \vee \SDP_{\A,a,\B,C_3} (\vx_{\start,p_1},\vx_{\start,p_2},\vx_{\start,p_3},\vx_{\start,p_4},\vx_{\start,p_5})
    \\
    \vx_{\start,p_4} 
    &= 
    \neginfty \vee 1 \vee \SDP_{\A,a,\B,C_4} (\vx_{\start,p_1},\vx_{\start,p_2},\vx_{\start,p_3},\vx_{\start,p_4},\vx_{\start,p_5})
    \\
    \vx_{\start,p_5} 
    &= 
    \neginfty \vee 7 \vee \SDP_{\A,a,\B,C_5} (\vx_{\start,p_1},\vx_{\start,p_2},\vx_{\start,p_3},\vx_{\start,p_4},\vx_{\start,p_5})
  \end{align}
  
  \newcommand{\eins}{\begin{pmatrix}1 & 0 & 0 \\ 0 & 0 & 0 \\ 0& 0 & 0 \end{pmatrix}}

  \noindent
  Here
  \begin{center}\scalebox{1}{$
      \A = \left( \eins \right) 
      \qquad
       a = (1) 
  $}\end{center}
\vspace*{-3mm}
  \begin{center}$
      \B =
        \left( 
          \begin{pmatrix} 0 & -0.5 & 0 \\ -0.5 & 0 & 0 \\ 0 & 0 & 0 \end{pmatrix}
          ,
          \begin{pmatrix} 0 &  0.5 & 0 \\  0.5 & 0 & 0 \\ 0 & 0 & 0 \end{pmatrix}
          ,
          \begin{pmatrix} 0 & 0 & -0.5 \\ 0 & 0 & 0 \\ -0.5 & 0 & 0 \end{pmatrix}
          ,
          \begin{pmatrix} 0 &  0 & 0.5 \\  0 & 0 & 0 \\ 0.5 & 0 & 0 \end{pmatrix}
          ,
          \begin{pmatrix} 0 &  0 & 0 \\  0 & 2 & 1 \\ 0 & 1 & 3 \end{pmatrix}
        \right)
  $\end{center}
\vspace*{-3mm}
  \begin{align*}
       C_1 &= \begin{pmatrix} 0 & -0.5 & -0.005 \\ -0.5 & 0 & 0 \\ -0.005 & 0 & 0 \end{pmatrix} 
       &
       C_2 &= \begin{pmatrix} 0 &  0.5 &  0.005 \\  0.5 & 0 & 0 \\  0.005 & 0 & 0 \end{pmatrix} 
       \\
       C_3 &= \begin{pmatrix} 0 & 0.005 & -0.495 \\ 0.005 & 0 & 0 \\ -0.495 & 0 & 0 \end{pmatrix} 
       &
       C_4 &= \begin{pmatrix} 0 & -0.005 & 0.495 \\ -0.005 & 0 & 0 \\ 0.495 & 0 & 0 \end{pmatrix}
       \\
       C_5 &= \begin{pmatrix} 0 & 0 & 0 \\ 0 & 1.9803 & 0.9802 \\ 0 & 0.9802 & 2.9603 \end{pmatrix} 
  \end{align*}
  %
  In this example we have $3^5 = 243$ different $\vee$-strategies.
  Assuming that the algorithm always chooses the best local improvement,
  in the first step it switch to the $\vee$-strategy that
  is given by the finite constants.
  At each equation, 
  it then can switch to the $\SDP$-expression, but 
  then, because it constructs a strictly increasing sequence,
  it can never return to the constant.
  Summarizing,
  because of the simple structure,
  it is clear that our $\vee$-strategy improvement algorithm will
  perform at most $6$ $\vee$-strategy improvement steps.
  In fact our prototypical implementation performs 
  $4$ $\vee$-strategy improvement steps on this example.
  \qed
\end{example}

%% file: conc.tex
\np
\section{Conclusion}
\label{s:conc}

\ok
We introduced and studied systems of \maxmorcave equations ---
a natural and strict generalization of systems of rational equations 
that were previously studied by \citet{DBLP:conf/csl/GawlitzaS07,DBLP:journals/toplas/GawlitzaS11}.
We showed how the $\vee$-strategy improvement approach
from \citet{DBLP:conf/esop/GawlitzaS07,DBLP:conf/csl/GawlitzaS07} 
can be generalized to solve these fixpoint equation systems.
We provided full proves and a in-depth discussion on the different cases.

On the practical side,
we showed that our algorithm 
can be applied to perform static program analysis 
w.r.t.\ quadratic templates 
using the relaxed abstract semantics of \citet{DBLP:conf/esop/AdjeGG10}
(based on Shor's semi-definite relaxation schema).
This analysis can, for instance, be used
to verify linear recursive filters and numerical integration schemes.
In the conference article that appears in the proceedings of 
the Seventeenth International Static Analysis Symposium (SAS 2010)
we report on experimental results that were obtained through our proof-of-concept implementation
\cite{DBLP:conf/sas/GawlitzaS10}.

For future work, we are interested in 
studying the use of other convex relaxation schemes 
to deal with more sophisticated cases, 
a problem already posed by \citet{DBLP:conf/esop/AdjeGG10}.
This would partially abolish the restriction to affine assignments and quadratic guards. 
Currently,
we apply our $\vee$-strategy improvement algorithm only to numerical static analysis of programs.
It remains to investigate 
in how far the $\vee$-strategy improvement algorithm we developed can be applied to other applications
--- maybe in other fields of computer science. 
Since our methods are solving quite general fixpoint problems,
we have some hope that this is the case.
Natural candidates could perhaps be found in the context of two-players zero-sum games.
